\documentclass[aps,pra,reprint,a4paper,superscriptaddress,floatfix,longbibliography]{revtex4-2}
\usepackage[british]{babel}
\usepackage{graphicx}
\usepackage{epsfig}
\usepackage{microtype}
\usepackage{amsfonts}
\usepackage{amsmath}
\usepackage{amssymb}
\usepackage{amsthm}
\usepackage{enumitem}
\usepackage{footmisc}
\usepackage{color}
\usepackage{tikz}
\usepackage{scalerel}
\usepackage[colorlinks=true,linkcolor=blue,citecolor=magenta,urlcolor=blue]{hyperref}
\usepackage{comment}
\usepackage{braket}
\usepackage{physics}
\usepackage{bm}

\theoremstyle{plain}
\newtheorem{theorem}{Theorem}
\newtheorem{proposition}[theorem]{Proposition}

\theoremstyle{definition}

\theoremstyle{remark}

\newtheorem*{remark*}{Remark}

\newcommand{\scenS}{\mathbf{S}}
\newcommand{\scenT}{\mathbf{T}}
\newcommand{\scenR}{\mathbf{R}}

\newcommand{\C}{\mathcal{C}}
\newcommand{\CS}{\C_{\scenS}}
\newcommand{\CT}{\C_{\scenT}}
\newcommand{\X}{\mathcal{X}}
\newcommand{\XS}{\X_{\scenS}}

\newcommand{\OO}[1]{\mathcal{O}_{#1}}
\newcommand{\OS}[1]{\mathcal{O}_{\scenS,#1}}
\newcommand{\OT}[1]{\mathcal{O}_{\scenT,#1}}
\newcommand{\OX}{\OO{\X}}
\newcommand{\OA}{\OO{A}}
\newcommand{\OSA}{\OS{A}}
\newcommand{\OTA}{\OT{A}}

\newcommand{\old}{\mathsf{old}}
\newcommand{\new}{\mathsf{new}}
\newcommand{\mc}{\mathsf{mc}}
\newcommand{\pc}{\mathsf{pc}}
\newcommand{\oldA}{\mathsf{old}{\scaleobj{0.8}{{+}}}\!A}
\newcommand{\mcA}{\mathsf{mc}{\scaleobj{0.8}{{+}}}\!A}
\newcommand{\pcA}{\mathsf{pc}{\scaleobj{0.8}{{+}}}\!A}

\newcommand{\Cold}{\mathcal{C}_{\old}}
\newcommand{\Cnew}{\mathcal{C}_{\new}}
\newcommand{\ColdA}{\mathcal{C}_{\oldA}}
\newcommand{\CmcA}{\mathcal{C}_{\mcA}}
\newcommand{\Cmc}{\mathcal{C}_{\mc}}
\newcommand{\CpcA}{\mathcal{C}_{\pcA}}

\newcommand{\vold}{v_{\old}}

\newcommand{\vmc}{v_{\mc}}
\newcommand{\vmcA}{v_{\mcA}}
\newcommand{\vpcA}{v_{\pcA}}
\newcommand{\Vold}{V_{\old}}

\newcommand{\VmcA}{V_{\mcA}}

\newcommand{\qold}{q_{\old}} 
\newcommand{\qmcA}{q_{\mcA}} 
\newcommand{\qpcA}{q_{\pcA}} 

\newcommand{\pold}{p_{\old}}
\newcommand{\pmc}{p_{\mc}}

\newcommand{\ICHSH}{I_{\mathsf{CHSH}}}

\newcommand{\LNC}{L_{N\!C}}
\newcommand{\LQ}{L_{Q}}
\newcommand{\LND}{L_{N\!D}}

\newcommand{\lambdachange}[2]{\lambda[#1 \mapsto #2]}

\newcommand{\ItrAZ}{{I, \tr(A),0}}
\newcommand{\Zakak}{{0, a_k, a_k}}
\newcommand{\bItrAZ}{b_{\ItrAZ}}
\newcommand{\bZakak}{b_{\Zakak}}

\begin{document}


\title{Lifting noncontextuality inequalities}


\newcommand{\inllong}{INL -- International Iberian Nanotechnology Laboratory, Av.~Mestre Jos\'e Veiga s/n, 4715-330 Braga, Portugal}
\newcommand{\inlshort}{INL -- International Iberian Nanotechnology Laboratory, Braga, Portugal}
\newcommand{\haslablong}{HASLab, INESC TEC, Universidade do Minho, Departamento de Informática, Campus de Gualtar, 4710-057 Braga, Portugal}
\newcommand{\haslabshort}{HASLab, INESC TEC, Universidade do Minho, Braga, Portugal}
\newcommand{\diumlong}{Departamento de Inform\'atica, Universidade do Minho, Campus de Gualtar, 4710-057 Braga, Portugal}
\newcommand{\diumshort}{Departamento de Inform\'atica, Universidade do Minho, Braga, Portugal}

\newcommand{\dfaIIlong}{Departamento de F\'{\i}sica Aplicada II, Universidad de Sevilla, ETS Arquitectura, Avda. Reina Mercedes 2, 41012 Sevilla, Spain}
\newcommand{\dfaIIshort}{Departamento de F\'{\i}sica Aplicada II, Universidad de Sevilla, Sevilla, Spain}
\newcommand{\icIlong}{Instituto Carlos~I de F\'{\i}sica Te\'orica y Computacional, Universidad de Sevilla, Facultad de Física, Av. Reina Mercedes s/n, 41012 Sevilla, Spain}
\newcommand{\icIshort}{Instituto Carlos~I de F\'{\i}sica Te\'orica y Computacional, Universidad de Sevilla, Sevilla, Spain}
 
\author{Raman Choudhary}
\affiliation{\inlshort}
\affiliation{\haslabshort}
\affiliation{\diumshort}

\author{Rui Soares Barbosa}
\affiliation{\inlshort}

\author{Ad\'an Cabello}
\affiliation{\dfaIIshort}
\affiliation{\icIshort}


\begin{abstract}
Kochen--Specker contextuality is a fundamental feature of quantum mechanics and a crucial resource for quantum computational advantage and reduction of communication complexity.
Its presence is witnessed in empirical data by the violation of noncontextuality inequalities.
However, all known noncontextuality inequalities corresponding to facets of noncontextual polytopes are either Bell inequalities or refer to cyclic or state-independent contextuality scenarios. 
We introduce a general method for lifting noncontextuality inequalities, deriving facets of noncontextual polytopes for more complex scenarios from known facets of simpler sub-scenarios.
Concretely, starting from an arbitrary scenario, the addition of a new measurement or a new outcome preserves the facet-defining nature of any noncontextuality inequality.
This extends the results of Pironio [J. Math. Phys. \textbf{46}, 062112 (2005)] from Bell nonlocality scenarios to contextuality scenarios, unifying liftings of Bell and noncontextuality inequalities.
Our method produces facet-defining noncontextuality inequalities in all scenarios with contextual correlations, and we present examples of facet-defining noncontextuality inequalities for scenarios where no examples were known.
Our results shed light on the structure of noncontextuality polytopes and the relationship between such polytopes across different scenarios.
\end{abstract}


\maketitle


\section{Introduction}


\subsection{Motivation}


Kochen--Specker (KS) contextuality \cite{Kochen:1967JMM,Budroni:2022RMP}, i.e.\ the impossibility of explaining with a single global probability distribution the marginal probability distributions produced by either ideal measurements of compatible observables \cite{Klyachko:2008PRL,Cabello:2008PRL} or by measurements on spatially separated systems \cite{Bell:1964PHY,Clauser:1969PRL}, is a characteristic signature of quantum mechanics.
More recently, it has been shown that contextuality is a crucial resource for universality in some models of quantum computation \cite{Howard:2014NAT,Raussendorf:2013PRA}, non-oracular quantum computational advantage \cite{Bravyi:2018SCI}, quantum communication complexity advantage \cite{Gupta:2023PRL}, and secure communication \cite{Horodecki:2010XXX,Zhen:2023PRL}.

Contextuality is typically witnessed by the violation of linear constraints called noncontextuality (NC) inequalities.
Indeed, the set of noncontextual correlations for a given measurement scenario forms a convex polytope.
A minimal set of conditions for deciding whether measurement statistics are contextual is provided by the set of inequalities that support its facets.
However, all known facet-defining NC inequalities
are either for Bell scenarios \cite{Clauser:1969PRL} or refer to cyclic \cite{Klyachko:2008PRL} or state-independent \cite{Cabello:2008PRL,Badziag:2009PRL,Yu:2012PRL} contextuality scenarios.

Moreover, characterising all the facet-defining inequalities of such polytopes
is a notoriously hard problem (NP-complete \cite{Avis_2005}) to solve in general.
As such, there are few fully characterised scenarios, for which all facet-defining NC inequalities are known:
non-Bell scenarios include two-outcome $k$-cycle ($k \geq 5$) scenarios \cite{PhysRevA.88.022118},
Bell scenarios include the $(2,2,2)$ CHSH scenario \cite{Clauser:1969PRL}, and various classes generalising it such as $(2,m,2)$ \cite{PhysRevD.2.1418,BraunsteinAP1990}, $(n,2,2)$ \cite{PhysRevD.35.3066}, and $(2,2,k)$ \cite{Collins:PRL2002} scenarios,
where $(n,m,k)$ stands for $n$ parties, each with $m$ measurement settings, each with $k$ possible outcomes. 
Furthermore, there are many Bell scenarios for which a partial characterisation of their Bell inequalities has been carried out.

Despite the demoralising hardness of characterising arbitrary scenarios, some non-trivial work has been done linking simpler Bell scenarios with more complex ones.
In Ref.~\cite{pironio2005lifting}, Pironio proposed a method to derive (some of the) facet-defining Bell inequalities of complex Bell polytopes
starting from known inequalities of simpler Bell polytopes.
It employed the idea of \emph{lifting}, a commonly used technique in convex polyhedral theory to derive facet-defining inequalities of a polytope in $\mathbb{R}^n$ from facet-defining inequalities of a related polytope in $\mathbb{R}^{m}$ where $m < n$.
The upshot is that once the facets of a simpler polytope have been fully or partially identified,
one need not start from scratch when searching for the facets of a more complex polytope.
One may instead concentrate efforts on finding the facet-defining inequalities that are absent in or do not arise from simpler scenarios.

Aiming to foster facet characterisation for general Bell scenarios,
Pironio showed that any facet-defining inequality of an arbitrary Bell polytope can be lifted to one or more facet-defining inequalities of any more complex Bell polytope, where by `more complex' we mean a Bell scenario with more parties, more local measurements for a party, or more outcomes for a measurement (or a combination of all three) than the original scenario.
Building on this work, in Ref.~\cite{Pironio_2014} Pironio characterised Bell scenarios whose only facets are given by liftings of the CHSH inequality.
These include e.g. the bipartite scenarios where one party has a binary choice of dichotomic measurements, irrespective of the number of measurement settings and outcomes for the other party.

Non-Bell-type contextuality scenarios, on the other hand, have not received as much attention in terms of facet characterisation.
This work aims to address this gap.

\subsection{Contributions}
\label{sec:contributions}

We introduce a method for producing facet-defining NC inequalities in arbitrary KS contextuality scenarios.
This is based on -- and strictly extends -- the lifting techniques used by Pironio \cite{pironio2005lifting} for Bell scenarios.
The method allows us to identify facet-defining NC inequalities for \emph{all} scenarios which admit contextual correlations and thus provides a key to explore an infinity of as-of-yet unexplored scenarios.
This is ensured by Vorob{\textquotesingle}ev's theorem \cite{Vorobyev:1962}, which guarantees that any contextuality-witnessing scenario contains an induced $k$-cycle sub-scenario (for some $k \geq 4$), and by the complete characterisation of the noncontextual polytopes for all such cycle scenarios \cite{PhysRevA.88.022118}.

We now give a concise, high-level summary of our main results,
bearing in mind that the concepts involved will only be properly introduced in later sections. 
A scenario $\scenS$ can be extended to a more complex scenario $\scenT$ by adding more measurements and/or more outcomes.
We focus on one such step at a time:
either adding \emph{one} measurement or adding \emph{one} outcome for an already existing measurement.
An arbitrary extension can be seen as a sequence of such single-step extensions.
We fix an initial facet-defining NC inequality for $\scenS$
which we aim to lift to (one or more) facet-defining NC inequalities for $\scenT$.
We achieve this in both situations, but the specifics differ somewhat.

\paragraph*{Measurement lifting.}
When $\scenT$ is obtained from $\scenS$ by adjoining a new measurement $A$,
the method of lifting depends both
on the compatibility relations between $A$ and the pre-existing measurements
and on the facet-defining inequality being lifted, namely on the set of measurements that \emph{effectively contribute} to that inequality.
We distinguish two cases:
\begin{enumerate}[leftmargin=*,align=left,label=\Roman*.,widest*=2]
   \item
   If $A$ is incompatible with some measurement that contributes to the initial inequality, then the inequality lifts unchanged to a facet-defining inequality for the scenario $\scenT$. The new measurement $A$ is traced out and does not effectively contribute to the lifted inequality, thus it need not be performed for testing the inequality.
  \item
  If $A$ is compatible with all the measurements that contribute to the initial inequality, then the inequality lifts to
  a facet-defining inequality of $\scenT$ for each outcome $a_k$ of $A$. Each such lifted inequality can be tested by first measuring $A$, post-selecting on obtaining the outcome $a_k$, and then testing the initial inequality.
\end{enumerate}

\paragraph*{Outcome lifting.}
When $\scenT$ is obtained from $\scenS$ by adding a new outcome $a_0$ for an already existing measurement $A$,
the original inequality lifts to a facet-defining inequality on $\scenT$ for each
choice of (pre-existing) outcome $a_k \neq a_0$ of $A$. 
The choice indicates the outcome $a_k$ of $A$ with which $a_0$ is to be `clubbed together'.
When testing the inequality, any occurrence of the new outcome $a_0$ for $A$ is treated as if it were an $a_k$ in the original scenario.
An exception is that the initial inequality cannot itself be the result of case II measurement lifting with the \textit{same} choice of outcome $a_k$: intuitively, that would entail post-selecting on two different outcomes for $A$.


\subsection{Structure}


In Sec.~\ref{sec:concepts}, we introduce all the definitions required in the work.
The first three subsections present the framework for studying contextuality: measurement scenarios (\ref{sec:scenarios}), correlations (\ref{sec:correlations}), and the noncontextual polytope (\ref{sec:NCpolytope}).
Sec.~\ref{sec:polytope} contains concepts of polytope theory.

In Sec.~\ref{sec:overview}, we give an informal overview of our lifting techniques.
Sec.~\ref{sec:Visualising Lifting} illustrates the gist of our liftings via some trivial but visualisable polytopes.
Sec.~\ref{sec:KS vs Bell} compares our work with Pironio's, describing the shift in perspective required for our extension of his lifting methods from Bell to NC inequalities.

In Sec.~\ref{sec:results}, we formally present our results.
After fixing notation in Sec.~\ref{sec:notation}, 
Sec.~\ref{sec:Proof outline} briefly outlines the logical structure common to the various lifting proofs presented.
Then, Secs.~\ref{sec:Measurement lifting} and~\ref{sec:outcome lifting} describe (and prove) in detail our methods for measurement lifting and outcome lifting, respectively.

In Sec.~\ref{sec:applications}, we employ our results to obtain facet-defining NC inequalities for scenarios for which no such inequalities were known.
Sec.~\ref{sec:App1} applies sequential measurement lifting to the CHSH inequality to obtain facet-defining NC inequalities for the two-outcome anti-heptagon scenario $\Bar{C}_7$ (the complement of the $7$-cycle),
while
Sec.~\ref{sec:App2} applies sequential outcome lifting to the KCBS inequality to obtain facet-defining NC inequalities for the three-outcome $5$-cycle scenario $C_5$, the simplest non-Bell cyclic scenario with three outcomes per measurement.
Sec.~\ref{sec:AppDiscussion} wraps up with general remarks concerning such applications.

In Sec.~\ref{sec:discussion}, we conclude and suggest some directions for future research.


\section{Concepts}
\label{sec:concepts}


In this section, we introduce the necessary background, including the setup to discuss contextuality and elements of polytope theory.

\subsection{Measurement scenarios}
\label{sec:scenarios}


Abstractly, a \emph{measurement scenario} (a.k.a.\ contextuality scenario or simply scenario) $\scenS$ is given by: a finite set $\XS$ of measurements, a symmetric reflexive relation on $\XS$ indicating \emph{compatibility}, and a finite set $\OSA$ of outcomes for each measurement $A \in \XS$.
We think of the set $\XS$ equipped with the compatibility relation as a graph, known as the \emph{compatibility graph} of the scenario.
A \emph{measurement context} is a clique in the compatibility graph, i.e.\ a set of pairwise-compatible measurements.
We write $\CS = \{C_1, \ldots, C_n\}$ for the set of maximal contexts, fixing an enumeration.
By maximal we mean that for all $C_i$ and $C_j$ in $\CS$, $C_i \subseteq C_j$ implies $C_i = C_j$.
We omit the subscript $\scenS$ whenever discussing a single scenario, as in the remainder of this section.

A \emph{Bell scenario} is a scenario in which measurements are assigned to parties, and where any two measurements are compatible if and only if they belong to different parties.
So, the compatibility graph is a complete $N$-partite graph where $N$ is the number of parties, and the maximal contexts correspond precisely to a choice of one measurement per party. This is the setting appropriate to study Bell nonlocality.

In a quantum mechanical realisation of a scenario, measurements are taken to be quantum observables and compatibility has the usual meaning of joint measurability. When discussing KS contextuality, measurements are assumed to be ideal, i.e.~yielding the same outcome when repeated and not disturbing any compatible observable.
Therefore, in quantum theory, they are represented as PVMs and compatibility corresponds to commutativity of the corresponding operators.

In a Bell scenario, one does not typically assume that the measurements must be ideal, so that in quantum theory they could be represented as POVMs.
On the other hand, each party's measurements must be local observables acting on a part of the system that is spatially separated from all the others; mathematically, on a factor of a tensor product Hilbert space.
It turns out, however, that all quantum-mechanical correlations in Bell scenarios are achievable by ideal measurements.
Consequently, not only do the notions of classicality agree, but also the set of quantum correlations for a Bell scenario (interpreted as a nonlocality scenario) and the set of quantum correlations for the corresponding contextuality scenario are identical.
This observation justifies regarding nonlocality as a special case of contextuality \cite{CSWPRL2014,Abramsky:2011NJP,AcinCMP2015}.
It has proved very useful in understanding the principles that shape quantum Bell nonlocal correlations \cite{Cabello:2013PRL,Cabello:2019PRA}.


\subsection{Correlations}
\label{sec:correlations}


A \emph{correlation} (a.k.a.\ empirical model or behaviour) on a scenario $\scenS$
collects specific outcome statistics for that scenario.
It is a family comprising a probability distribution on the joint outcomes for each maximal context in $\C$.
For each $C_i$, we fix an enumeration of the set of joint outcomes $\OO{C_i} = \prod_{A \in C_i} \OA = \{s^{1}_i, s^{2}_i,\ldots, s^{m_i}_i\}$.
A correlation can then be represented as a vector
\begin{equation*}
p  
=
[p(s^{1}_1),\ldots,p(s^{m_1}_1),\ldots,p(s^{1}_n),\ldots,p(s^{m_n}_n)],
\end{equation*}
where the component $p(s^{j}_i)$, often written more explicitly as $p(s^{j}_i \mid C_i)$,
represents the probability of observing the $j${-th} joint outcome $s^j_i$ upon jointly performing the measurements of the $i${-th} context $C_i$.
We have $p \in \mathbb{R}^d$, where $d = \sum_{C_i \in \C}|\OO{C_i}|$ is the dimension of the vector space wherein all the correlations on the given scenario live.

The fact that they determine a probability distribution $p_i$ over each maximal context $C_i$
means that the components must be all non-negative, i.e.\ $p(s^{j}_i) \geq 0$ for all $i,j$,
and in addition satisfy normalisation for each context $C_i$,
\begin{equation*}
 \sum_{j=1}^{m_i}p(s_i^{j})=1. \label{eq:8}
\end{equation*}

Moreover, we assume that the marginal probability distributions on the outcomes of each observable (or set of observables) is independent of the context, i.e.\ of what other measurements are performed together.
This condition is called \emph{no-disturbance} -- or, in Bell scenarios, \emph{no-signalling} -- and it prevents using the statistics for achieving communication between observers that are performing compatible measurements.
Correlations obtainable in quantum mechanics, in the ways briefly outlined above, always satisfy this condition.
More precisely, no-disturbance can be stated as follows: for any pair of maximal contexts $C_i$ and $C_{i'}$, the corresponding distributions $p_i$ and $p_{i'}$ determine the same marginal distribution over the intersection $C_i \cap C_{i'}$, i.e. 
for each (partial) joint outcome $t \in \OO{C_i \cap C_{i'}}$, it holds that
\begin{equation}
 \sum_j p(s_{i}^j) \;=\; \sum_{j'} p(s_{i'}^{j'}), \label{eq:nodisturbance}
\end{equation}
where the index $j$ runs over the set $\{j \mid s_{i}^j|_{C_i\cap C_{i'}}=t\}$ indexing those joint outcomes in $\OO{C_i}$ that extend $t$, and analogously for $j'$.

One may think of either side of Eq.~\eqref{eq:nodisturbance} as determining the probability $p(t)$ of obtaining outcome $t$ in the partial context $C_i \cap C_{i'}$; no-disturbance guarantees it is uniquely well defined.
Notice that the no-disturbance assumption justifies considering only maximal contexts in $\C$.


\subsection{The noncontextual polytope}
\label{sec:NCpolytope}


We are interested in characterising the set of correlations that admit a classical, i.e.\ noncontextual, explanation.

A \emph{noncontextual deterministic assignment} is a joint assignment of outcomes to all the measurements in a scenario, $\lambda \in \OO{\X} = \prod_{A \in \X} \OA$, which one may think of as a \emph{hidden variable}.
It determines a deterministic correlation
\begin{equation*}
v_{\lambda}= [v_\lambda(s^{1}_1),\ldots,v_\lambda(s^{m_1}_1),\ldots,v_\lambda(s^1_n),\ldots,v_\lambda(s^{m_n}_n)],
\end{equation*}
where $v_\lambda(s^{j}_i) \in \{0,1\}$ indicates whether $\lambda$ assigns the joint outcome $s^j_i$ to the context $C_i$,
\[
v_\lambda(s^{j}_i) = \begin{cases} 1 & \text{if $\lambda|_{C_i} = s^j_i$} \\ 0 & \text{otherwise.} \end{cases}
\]
Note that over each context $C_i$ such a vector $v_\lambda$ has precisely one component valued $1$ and the remaining $0$.

A correlation $p$ is said to be \emph{noncontextual (NC)} if it can be described as a convex combination of such deterministic correlations, i.e.\ if
\begin{equation}
 p=\sum_{\lambda \in \OX}\alpha_\lambda v_{\lambda} \label{eq:nc-convexvertices}
\end{equation}
for some $\alpha_\lambda \geq 0$ with $\sum_\lambda \alpha_\lambda = 1$.
This means that each component satisfies
\begin{equation*}
    p(s^{j}_i)= \sum_\lambda \alpha_\lambda v_\lambda(s^{j}_i). 
\end{equation*}
One may think of the coefficients $\alpha_\lambda$ as defining a probability distribution on the hidden-variable space of global assignments.

Eq.~\eqref{eq:nc-convexvertices} means that the set of NC correlations is the convex hull of a finite number of vectors $\{v_{\lambda}\}_{\lambda \in \OX}$.
It thus forms a convex polytope in $\mathbb{R}^d$, here given in terms of its \emph{V-representation}; see Sec.~\ref{sec:polytope} for more details on convex polytopes.
Any convex polytope can equally be described as the intersection of a finite number of half-spaces, given by linear inequalities; this is called an \emph{H-representation} of the polytope. 
Each valid inequality for the noncontextual polytope of a given scenario is known as a \emph{noncontextuality (NC) inequality} for that scenario, marking a boundary between the noncontextual and contextual regions.
The violation of any such inequality \cite{Klyachko:2008PRL,Cabello:2008PRL,Budroni:2022RMP} by a correlation witnesses contextuality, the nonexistence of a global probability distribution (on outcome assignments for all measurements in $\X$) whose marginals reproduce the given correlation. 

A minimal H-representation is given by the facet-defining inequalities; see Sec.~\ref{sec:polytope} below for details.
Although for restricted classes of polytopes, such as simplicial polytopes, there are known polynomial-time algorithms for facet enumeration (the problem of computing the H-representation given the V-representation) \cite{avis1991pivoting,10.1145/321556.321564,SWART198517}, characterising all the facets of the noncontextuality polytope for an arbitrary scenario is a notoriously hard problem to solve (NP-complete \cite{Avis_2005}).


Given a measurement scenario, let $\LND$ denote the subset of $\mathbb{R}^d$ consisting of
all correlation vectors that satisfy no-disturbance.
Similarly, write $\LQ$ and $\LNC$ for the sets of quantum correlations and of noncontextual correlations, respectively.
Both $\LND$ and $\LNC$ are convex polytopes, whereas $\LQ$ is also convex but in general not a polytope.
There is a sequence of inclusions
\[
\LNC \;\subseteq\; \LQ \;\subseteq\; \LND,
\]
and moreover these inclusions are strict for any non-trivial scenario.
However, it holds that the affine dimension of all three sets is the same, 
\begin{equation*}
\dim(\LNC) \;=\; \dim(\LQ) \;=\; \dim(\LND).
\end{equation*}
This fact can be interpreted as stating that quasi-probability distributions (which are normalised but may take negative values)
on deterministic noncontextual assignments can account for any non-disturbing correlation \cite{Abramsky:2011NJP}.
Along with normalisation, the no-disturbance condition imposes linear equality constraints on the components of a correlation vector,
rendering $\LND$ (and thus also $\LNC$) not full dimensional, i.e.\ $\dim(\LND) < d$.
In general, such linear constraints must be accounted for in the minimal H-representation of any non-full-dimensional polytope, as they introduce a degree of non-uniqueness to the form of the facet-defining inequalities.

In this work, our focus is on characterising the polytope $\LNC$ of noncontextual correlations for arbitrary scenarios.
To simplify the presentation, given a scenario $\scenS$, we denote its noncontextual polytope by $S$ (analogously, $\scenT$ and $T$).


\subsection{Convex polytopes}
\label{sec:polytope}


We recall some basic definitions of polytope theory that will be useful later on.
For more details, we refer the reader to \cite{Ziegler:2012}.
For brevity, in the context of this paper, the term polytope always means convex polytope.

An \emph{affine combination} of a set of points $p_1,\ldots, p_k \in \mathbb{R}^d$ is a linear combination whose coefficients add up to $1$,
i.e.\ $\sum_{i=1}^{k} \alpha_i p_i$ for some $\alpha_i \in \mathbb{R}$ with $\sum_i \alpha_i = 1$.
A \emph{convex combination} is an affine combination with the additional condition that all coefficients be non-negative, i.e.\ $\alpha_i \geq 0$. The \emph{affine span} (\emph{convex hull}) of a set of points is the set of all their affine (convex) combinations.

A set of points $p_0,\ldots, p_k \in \mathbb{R}^d$ is \emph{affinely independent}
when $\sum_{i=0}^{k} \alpha_i p_i= 0$ with $\alpha_i \in \mathbb{R}$ satisfying $\sum_i \alpha_i = 0$ implies that $\alpha_i = 0$ for all $i$.
This is equivalent to the vectors $p_1 - p_0, \ldots, p_k - p_0$ being linearly independent.
It is also equivalent to saying that no element of the set is an affine combination of the others.

A set of $k+1$ affinely independent points spans a set of dimension $k$.
More generally, the affine span of a set of points has dimension $k$ if and only if the maximum number of affinely independent points in that set is $k+1$.
The affine span of any such $k+1$ points equals the affine span of the whole set, and those $k+1$ points are called an \emph{affine basis} for the set.
The dimension $\dim(S)$ of an arbitrary subset $S$ of $\mathbb{R}^d$ is defined to be the dimension of its affine span.

A convex polytope is a subset $P \subseteq \mathbb{R}^d$ that is a convex hull of a finite set of points, its \emph{vertices}.
This set of vertices constitutes the V-representation of the polytope.

A \emph{linear inequality} on  $\mathbb{R}^d$ is a
predicate 
$I(p) \equiv b \cdot p \geq \beta$
on a variable vector $p \in \mathbb{R}^d$,
determined by a vector $b \in \mathbb{R}^d$ of coefficients and a bound $\beta \in \mathbb{R}$.
It defines a \emph{half-space} of $\mathbb{R}^d$:
$\{p \in \mathbb{R}^d \mid I(p)\} = \{ p \in \mathbb{R}^d \mid b \cdot p \geq \beta\}$.
A linear inequality is a \emph{valid inequality} for a polytope $P \subseteq \mathbb{R}^d$
whenever
$P$ lies entirely within the half-space described by it, that is, when
\begin{equation*}
    P \cap \{p \in \mathbb{R}^d \mid b \cdot p \geq \beta\} = P.
\end{equation*}
Since any point in $P$ can be written as a convex combination of its vertices, it is sufficient to test an inequality on the vertices to check for its validity.

Given a valid inequality $b \cdot x \geq \beta$ for a polytope $P$, the set $F = \{p \in P \mid b \cdot p=\beta\}$ of points of $P$ that saturate the inequality is called the \emph{face of $P$ supported by the inequality}.
Each face of a convex polytope $P$ is a convex polytope itself. 
A face $F$ is called \emph{proper} when it is neither the empty face nor the whole polytope itself, i.e.\ when $F\neq \varnothing$ and $F\neq P$.
Clearly, for a proper face $F$ of $P$, $0 \leq \dim(F) < \dim(P)$.
The minimal proper faces, of dimension $0$, are the vertices of $P$.
A maximal proper face, of dimension $\dim(P) - 1$, is called a \emph{facet} of $P$.

According to the Minkowski--Weyl theorem \cite{Minkowski:1897,Weyl:1934,Ziegler:2012}, a convex polytope can equivalently be described as an intersection of finitely many half-spaces. This defines an \emph{H-representation} of the polytope,
\begin{equation*}
 P = \bigcap_{i} \{p \in \mathbb{R}^d \mid b^{(i)} \cdot p \geq \beta^{(i)} \}. 
\end{equation*}
The set of inequalities supporting the facets of $P$ gives a minimal complete description of the polytope as an H-representation.
In fact, any valid inequality of $P$ can be derived from the facet-defining inequalities.

By the remarks above about the dimension of any subset of $\mathbb{R}^d$, the maximum number of affinely independent points in a polytope $P$ is $\dim(P)+1$.
One can always draw such $\dim(P)+1$ affinely independent points from the vertices of $P$.
Similarly, any face $F$ of $P$ contains $\dim(F)+1$ affinely independent vertices of $P$, which of course saturate the inequality supporting $F$.
In particular, a facet contains a set of $\dim(P)$ affinely independent vertices of $P$.
This means that one can always choose an affine basis of $P$ comprising $\dim(P) + 1$ affinely independent vertices
in such a way that $\dim(P)$ of those vertices belong to a chosen facet $F$, and the remaining vertex can be chosen arbitrarily among the vertices not belonging to $F$.
We make use of this fact repeatedly in our proofs.


\section{Overview of lifting method}
\label{sec:overview}


Before diving into the detailed description of our results and proofs in Sec.~\ref{sec:results},
we offer an accessible, intuitive sketch of the main ingredients of our lifting method.
Sec.~\ref{sec:Visualising Lifting} uses small -- trivial yet visualisable -- examples to convey the flavour of the various forms of lifting.
In Sec~\ref{sec:KS vs Bell}, we discuss how these extend Pironio's method from Bell scenarios.

\subsection{Visualising lifting}
\label{sec:Visualising Lifting}
A scenario $\scenS$ can be extended to a larger scenario $\scenT$ through the addition of new measurements and/or outcomes.
In each case, $\scenS$ is a sub-scenario of $\scenT$.
It is natural to enquire whether the knowledge of a facet-defining NC inequality for the sub-scenario $\scenS$ gives us some information about facet-defining NC inequalities for $\scenT$.
Lifting answers this question in the affirmative in that it is a method to transform each facet-defining inequality of the noncontextual polytope $S$ into one or more facet-defining inequalities of the polytope $T$.

The noncontextual polytope for the simplest contextuality-witnessing scenario (i.e.\ a scenario for which $\LNC \neq \LND$) is $8$-dimensional.
This makes it impossible to visualise lifting starting from this polytope.
Nonetheless, there are some simple visualisable NC polytopes -- albeit for scenarios unable to witness contextuality -- which aptly capture the idea of lifting.
We use these examples to provide some intuition for the idea behind our lifting method.
Despite being thoroughly uninteresting from the point of view of contextuality,
these examples are indeed special cases of our lifting results which `contain all the germs of generality'.

\paragraph*{Measurement lifting}
For measurement lifting, we consider two cases, as outlined in Sec.~\ref{sec:contributions}.

{\em Case I}
applies when the new measurement added to $\scenS$ is not simultaneously compatible with all the measurements effectively contributing to the facet-defining inequality being lifted, i.e.\ it is incompatible with at least one such measurement;
see Sec.~\ref{sec:Measurement lifting} for details including the meaning of `effectively contributing'.
We illustrate this case with an example that serves as a proxy for all such situations.

Consider $\scenS$ the scenario with two dichotomic measurements, $A$ and $B$, that are incompatible with each other.
Its NC polytope $S$ is two-dimensional (embedded in a four-dimensional ambient vector space) and is shown on the left-hand side of Fig.~\ref{fig:fig1}(a).
Its four vertices correspond to the deterministic assignments -- $00$, $01$, $10$, $11$ -- to measurements $A$ and $B$, in that order.
Observe that the fact that $A$ and $B$ are incompatible induces an affine dependency among the vertices of the polytope,
\begin{equation}
v_{00} - v_{01} + v_{11} - v_{10} = 0 ,
\label{eq:affinedependence-incompatible}
\end{equation}
which would not hold were $A$ and $B$ compatible; cf.\ the simplex on the right-hand side of Fig.~\ref{fig:fig1}(b).
This kind of affine relation induced by incompatibility plays a crucial role in our proofs;
see for example Eq.~\eqref{eq:16}.

Since $S$ is a two-dimensional polytope, its facets are one-dimensional.
One of its facet-defining inequalities is depicted in the figure;
it supports the facet containing the vertices
$v_{10}$ and $v_{11}$. 

Now consider extending $\scenS$ to $\scenT$ by introducing a new measurement $C$ incompatible with both $A$ and $B$.
The NC polytope $T$ is three-dimensional (embedded in a six-dimensional ambient vector space), and is shown on the right-hand side of Fig.~\ref{fig:fig1}(a).
Notice that each outcome of $C$ determines an extension
of each deterministic assignment in $\scenS$ (hence, vertex of $S$)
to a deterministic assignment in $\scenT$ (vertex of $T$);
e.g.\ $v_{010}$ and $v_{011}$ in $T$ are the two extensions of $v_{01}$ in $S$.
One can think of this as captured by a polytope projection from $T$ to $S$ which `forgets' the outcome of $C$.

The polytope $T$ has a facet with vertices $v_{100}$, $v_{101}$, $v_{110}$, and $v_{111}$, whose supporting inequality is depicted on the right of Fig.~\ref{fig:fig1}(a).
These four vertices of $T$ are precisely the extensions of the vertices $v_{01}$ and $v_{11}$ of $S$, which saturate the initial inequality depicted on the left.
This inequality shown on the right of Fig.~\ref{fig:fig1}(a) is thus the \emph{lifting} of the inequality on the left.
One can similarly obtain three other facet-defining inequalities of $T$ from the remaining three facets of $S$.

In general, case I measurement lifting maps a facet $F$ of $S$ to the facet of $T$ whose set of vertices is exactly the set of all extensions of vertices in the original facet $F$.
As we will see in Sec.~\ref{sec:Measurement lifting}, the explicit form of the lifted inequality turns out to be, in a sense, the same as that of the initial inequality, since the outcome of the new measurement is ignored and thus `traced out'.



\begin{figure}
    \centering
    \includegraphics[scale = 0.75]{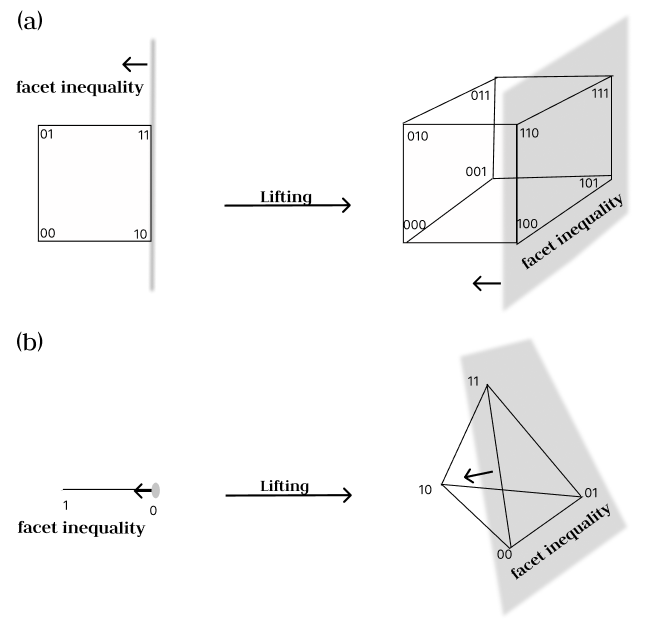}
    \caption{Depiction of measurement lifting in simple scenarios with dichotomic measurements.
    (a) Case I measurement lifting, from the scenario $\scenS$ with two incompatible measurements to the scenario $\scenT$ with three pairwise-incompatible measurements. The shaded line and directional arrow on the left represent a facet-defining inequality saturated by vertices $v_{10}$ and $v_{11}$, while the shaded plane and arrow on the right represent the lifted inequality, saturated by all extensions of $v_{11}$ and $v_{10}$ to $T$.
    (b) Case II measurement lifting, from the scenario with a single measurement to the scenario with two compatible measurements.
    The facet-defining inequality saturated only by $v_{0}$, depicted on the left, is lifted to a facet-defining inequality saturated by both extensions of $v_{0}$ ($v_{00}$ and $v_{01}$) plus one extension of $v_1$ ($v_{11}$), depicted on the right.}
    \label{fig:fig1}
\end{figure}


{\em Case II} applies when the newly added measurement is compatible with all the preexisting measurements that effectively contribute to the facet-defining inequality being lifted.

To visualise lifting in this case, we take $\scenS$ to be the scenario with a single dichotomic measurement, say $A$.
Its NC polytope $S$ is a one-dimensional line segment or $1$-simplex (embedded in a two-dimensional ambient vector space), shown on the left-hand side of Fig.~\ref{fig:fig1}(b).
It has only two facet-defining inequalities. One of them, saturated only by the vertex $v_0$, is depicted in the figure.

We extend this scenario to $\scenT$ by adding a dichotomic measurement $B$ compatible with $A$.
The resulting NC polytope $T$ is a three-dimensional tetrahedron or $3$-simplex (embedded in a four-dimensional ambient vector space), shown on the right-hand side of Fig.~\ref{fig:fig1}(b).

The facet-defining inequality depicted in the figure is one of two possible liftings of the initial inequality on the left of Fig.~\ref{fig:fig1}(b).
It is saturated by the vertices $v_{00}$, $v_{01}$, and $v_{11}$.
The first two are all the extensions of $v_0$, the vertex that saturates the initial inequality,
while the third is one of the two possible extensions of $v_{1}$, the other, non-saturating vertex of $S$. 
The vertex of $T$ corresponding to the other extension, $v_{10}$, does not saturate the inequality shown on the right of Fig.~\ref{fig:fig1}(b).
However, a different lifting of the same initial inequality shown on the left of Fig.~\ref{fig:fig1}(b) would support the facet containing the vertices $v_{00}$, $v_{01}$, and $v_{10}$.

In general, case II measurement lifting maps a facet $F$ of $S$ to a facet of $T$ whose set of vertices consists of:
(i) every possible extension of the vertices in the original facet $F$, and
(ii) every possible extension of the remaining vertices of $S$ (not in the original facet $F$)
except those that assign a chosen fixed outcome to the new measurement (outcome $0$ for $B$ in the example depicted above).
Varying the choice of this fixed outcome for the new measurement yields different liftings to $T$ of the same facet-defining inequality of $S$.

Notice that in both cases I and II above, the lifted inequality is effectively testing the original one, i.e.\ it simply represents a re-expression of the original inequality in the language of the larger polytope $T$.
Moreover, while in case I each facet of $S$ yields a single facet of $T$, in case II it yields one facet for each possible outcome of the newly added measurement.

\paragraph*{Outcome lifting}
In the case of outcome lifting, we consider extending a scenario $\scenS$ to $\scenT$ by adding an extra outcome $a_0$ to an existing measurement $A$.
Here, unlike in the case of measurement lifting, there is no (immediate) concept of extension of assignments, and thus of vertices of $S$ to vertices of $T$.
In fact, the vertices of $S$ could be seen as a strict subset of those of polytope $T$.

These vertices behave the same way with respect to an outcome-lifted inequality in $T$ as they do with respect to the initial facet-defining inequality of $S$, either saturating both or falling short by the same amount. However, there are more vertices in the larger polytope $T$, namely those that assign the new outcome $a_0$ to the measurement $A$.
Each such vertex behaves exactly as the vertex obtained by substituting $a_k$ for the outcome of $A$ (while leaving the rest of the assignment intact), for some fixed choice outcome $a_k$ of $A$ already present in the initial scenario $\scenS$.

In other words, a facet $F$ of $S$ is lifted to a facet of $T$ whose set of vertices consists of: the vertices in $F$ plus the vertices obtained from a vertex in $F$ by replacing (the fixed) outcome $a_k$ by $a_0$ for the measurement $A$.

As for case II measurement lifting, there is an element of choice involved.
Indeed, for each choice of $k$, one obtains a
(possibly different)
facet-defining inequality of $T$.
There is one exception, though: choosing $k$ fails to yield a facet of $T$ when the initial facet of $S$ is itself obtainable from a sub-scenario of $\scenS$ via case II measurement lifting by fixing the choice of outcome $a_k$ when adding measurement $A$;
in other words, when the facet-defining inequality being lifted is such that  $A$ is compatible with all the measurements effectively contributing to it and all its non-saturating vertices assign outcome $a_k$ to $A$.

One may think of each such choice of $k$ as determining a polytope projection $T \to S$ which performs a coarse-graining of the outcomes of $A$ by clubbing together $a_0$ and $a_k$ into the same outcome.
The lifted inequality then corresponds to `tracing out' along this identification, much as in case I measurement lifting.
This means that the form of the lifted inequality is such that it does not distinguish between the vertices with $a_k$ and $a_0$ as outcomes of $A$,
while the remaining vertices behave as they did with respect to the original inequality.
Indeed, in the new scenario, the lifted inequality is still in effect testing the original inequality.
That is why the new outcome needs to be `clubbed together' with -- and thus made indistinguishable from -- some pre-existing outcome, so that the `effective' number of outcomes is kept the same.


\subsection{Lifting: from Bell to KS}
\label{sec:KS vs Bell}


We now comment on the comparison with the work of Pironio in Ref.~\cite{pironio2005lifting} for the case of Bell scenarios, a source of inspiration for the present work.

From an arbitrary contextuality scenario, the simplest (single-step) extensions are of two types: 
(a) adding a new measurement, and (b) adding a new outcome for an already existing measurement.
Lifting across Bell scenarios, as in Ref.~\cite{pironio2005lifting}, means that the extended scenario must also be a Bell scenario.
In such extensions, type (a) can be further divided into two subcases:
adding a new party with a single measurement setting, or adding a new measurement to an already existing party.
When adding a new party with a single measurement, this new measurement is necessarily compatible with all the preexisting measurements (as they belong to different parties). Hence, this situation is always captured by case II measurement lifting (see Secs.~\ref{sec:contributions} and~\ref{sec:Visualising Lifting}).
On the other hand, when adding a new measurement to an already existing party, the new measurement is necessarily incompatible with all other measurements of that same party.
Thus, if the initial Bell inequality being lifted effectively involves this party, this situation falls within case I measurement lifting; otherwise, it falls within case II.

On a more technical note, the following distinction is relevant in establishing the Bell lifting results, and needs to be revised for our generalisation.
When adding a new party, the new measurement is appended to all previously existing maximal contexts.
By contrast, when adding a new measurement to an already existing party, the previous maximal contexts are preserved intact but new maximal contexts are introduced, namely those involving the newly added measurement.
When generalising to arbitrary contextuality scenarios, one no longer has a notion of parties.
Consequently, this neat separation is no longer possible,
and we are forced to consider situations in which a mix of new maximal contexts and extensions of previous maximal contexts arises.
This is reflected in the general form taken by a correlation vector on the extended scenario, shown in Eq.~\eqref{eq:form-vector-extendedspace} ahead.


\section{Lifting noncontextuality inequalities}
\label{sec:results}

We now move on to presenting our results in detail.
After fixing some notation in Sec.~\ref{sec:notation}, we sketch in Sec.~\ref{sec:Proof outline} the structure common to our proofs of measurement and outcome lifting.
We then delve into each of the forms of lifting in detail, in Secs.~\ref{sec:Measurement lifting} and~\ref{sec:outcome lifting}, respectively.

\subsection{Notation}
\label{sec:notation}


As before, when extending one scenario to another, the initial scenario is always denoted by $\scenS$ and its extension by $\scenT$. The respective noncontextual polytopes are denoted by $S$ and $T$.
The dimensions of these polytopes are denoted by $d_{S}$ and $d_T$, respectively.
The newly added measurement (for measurement lifting) or the existing measurement to which a new outcome is added (for outcome lifting) is denoted by $A$.

We use the letter $p$ to refer to a general vector in the space of correlations for $\scenS$, i.e.\ the ambient vector space of the polytope $S$, and we similarly use $q$ for $\scenT$.

We reserve the letters $u, v, w$ for vertices of $S$ or $T$.
Recall from Sec~\ref{sec:NCpolytope} that these vertices are of the form $v_\lambda$ corresponding to global assignments $\lambda$ of outcomes to all measurements.
To simplify the presentation, from now on we often elide the distinction between a vertex and its underlying assignment: e.g.\ we say that a vertex $v$ assigns a value to a measurement,
and write $v|_C$ for the joint outcome assignment over a context $C$ determined by the vertex $v$
(the restriction of the underlying global assignment to the measurements in $C$).

Recall that in specifying correlations on a scenario we only consider the maximal contexts, since statistics for their subsets can be recovered by marginalisation unambiguously due to no-disturbance.
We use the term \emph{partial context} to refer to a proper subset of a maximal context, whenever such emphasis is required.

Non-full-dimensional convex polytopes satisfy non-trivial affine equality constraints. In the case of NC polytopes, these correspond to the no-disturbance and normalisation conditions. 
A valid inequality $b \cdot p\geq \beta$ for a polytope remains valid
if one adds to it any linear combination of the equality constraints.
In the case of NC polytopes, one can always add a multiple of a normalisation condition to rewrite any NC inequality so that its lower bound is zero.
Therefore, without loss of generality, an arbitrary NC inequality will be represented in the form $b \cdot p \geq 0$.
Moreover, at various points,
we make use of the freedom afforded by the (linear) no-disturbance conditions to pick particularly convenient vectors of coefficients.

We typically denote by $B_P$ an affine basis of interest for a polytope $P$,
whose elements are taken from the vertices of the polytope.
We make crucial use of the fact that $\dim(P)$ vertices of $B_P$ can be chosen from any given facet.
When lifting a facet-defining inequality of $S$, we consider such an affine basis $B_S= \{u_{i}\}_{i=1}^{d_S + 1}$.
For slight convenience, we enumerate $B_S$ in a way that lists $d_S$ saturating vertices of the inequality (i.e.\ vertices in its supported facet) as the first $d_S$ elements, with the last element being the non-saturating one.


\subsection{Proof outline}
\label{sec:Proof outline}


Before burdening the reader with the mathematical details of the proofs of our results, we outline the logical structure underlying them. The proofs of measurement lifting and outcome lifting follow a common script.

Let $I_T$ be the inequality claimed to be the lifting of inequality $I_S$. To prove this claim we proceed as follows:

\begin{enumerate}[leftmargin=*,align=left]
\item We pick an affine basis $B_S = \{v_1, \ldots, v_{d_S+1}\}$ of vertices of $S$ chosen so that its first $d_S$ points saturate the initial inequality $I_S$.
 \item We use it to construct an over-complete affine basis $B_T$ of vertices of $T$, partitioned as $B_T = \bigcup_{j} V_j$,
 where each component $V_j = \{v^{(j)}_1, \ldots, v^{(j)}_{d_S + 1}\}$ consists of $d_{s}+1$ affinely independent vertices of $T$.
 Each $v^{(j)}_i$ is built from the corresponding vertex $v_i$ of $B_S$.
The specific construction depends on the type of lifting under consideration,
but the common aspect is that
$v^{(j)}_i$ is an `extension' of $v_i$ in that it is mapped to $v_i$ under the appropriate polytope projection $T \to S$.
 

\item By construction, the first $d_S$ vertices of each partition component $V_j$ saturate the inequality $I_T$.
Moreover, there is (at least) one partition component $V_k$ where the last vertex does not saturate $I_T$.

\item
For all other partition components, i.e.\ for all $j \neq k$,
either the last element $v^{(j)}_{d_S+1}$ saturates the inequality $I_T$,
or one can affinely eliminate it from $B_T$
by expressing it as an affine combination of the remaining vertices of $V_j$ and the vertices of $V_k$, essentially using Eq.~\eqref{eq:affinedependence-incompatible}.
This leaves only one remaining non-saturating vertex within $B_T$, namely  $v^{(k)}_{d_S+1}$, the last vertex in $V_k$.


 \item This is enough to establish the claim. Since (the reduced) $B_T$ affinely spans $T$, it contains a set of $d_T + 1$ affinely independent vertices. From the previous step, it has only one non-saturating vertex. It must therefore contain a set of $d_T$ affinely independent vertices that saturate $I_T$, proving that this inequality is facet-defining.

\end{enumerate}


\subsection{Measurement lifting}
\label{sec:Measurement lifting}


We first consider lifting under measurement extensions.
Let $\scenS$ be extended to $\scenT$ by the addition of a new measurement $A$.
Write $J \subseteq \XS$ for the set of all the measurements in $\scenS$ that are compatible with $A$, i.e.~the neighbourhood of $A$ in the compatibility graph of $\scenT$.

\subsubsection{The extended ambient space}


For any (not necessarily maximal) context, i.e.\ set of pairwise compatible measurements, $U$ contained in $J$, the set $U \cup \{A\}$ is a context of $\scenT$. If $U$ is maximal within $J$ then this results in a maximal context.
Thus, $A$ may contribute to maximal contexts of $\scenT$ in two ways:
\begin{enumerate}[leftmargin=*,align=left]
 \item create maximal contexts by being compatible with maximal contexts in $\CS$:
 for each context maximal within $J$ that is also maximal within $\XS$ (i.e. belongs to $\CS$), appending $A$ to it gives a (maximal) context in $\CT$. 
 \item create maximal contexts by being compatible with partial contexts in $\scenS$:
 for each context maximal within $J$ that is a partial context with respect to $\XS$ (i.e. does not belong to $\CS$), appending $A$ to it gives a (maximal) context in $\CT$. 
\end{enumerate}
Let $\CmcA$ denote the set of all contexts in $\CT$ arising as in item 1 above, where `$\mc$' reminds us that these contexts are extensions of maximal contexts of $\scenS$. Likewise, write $\CpcA$ for the set of contexts in $\CT$ arising as in item 2, where `$\pc$' indicates that they are extensions of partial contexts of $\scenS$.
Notice that the case covered by item 2 also implies the existence of maximal contexts in $\CS$ that are not fully compatible with $A$, and thus remain maximal contexts in $\CT$.
We write $\Cold$ for the set of all such contexts.
We can therefore write $\CT$ as a disjoint union of three mutually exclusive sets:
\begin{equation*}
 \CT = \Cold \cup \CmcA \cup \CpcA.
\end{equation*}
This induces a direct sum decomposition of the vector space of correlations for $\scenT$, which is the ambient space of the polytope $T$.
An arbitrary vector $q$ in this space can be written as 
\begin{equation}
 q = [\qold, \qmcA, \qpcA],
 \label{eq:form-vector-extendedspace}
\end{equation}
where $\qold$ represents the components indexed by joint outcomes of contexts in $\Cold$ and likewise for $\qmcA$ and $\qpcA$.

\subsubsection{Effectively contributing measurements}
\label{sec:effectivecontribution}

We need to consider two different forms of measurement lifting, depending on the
(compatibility behaviour of the) new measurement and the initial inequality being lifted.
The distinction hinges on the notion of an \emph{effectively contributing measurement to an inequality}, which we define by its negation.

Consider an inequality $b \cdot p \geq 0$ over the space of correlations for $\scenS$.
In short, a measurement $M \in \XS$ is said \emph{not} to contribute effectively to the inequality when the latter is insensitive to the outcome of $M$,
i.e.~its left-hand side
remains invariant under any change of outcome for $M$, rendering this measurement's role irrelevant to the inequality.

More formally, recall from Sec.~\ref{sec:NCpolytope} that the vertices of $S$ are deterministic noncontextual correlations $v_\lambda$, determined by global assignments $\lambda$ of outcomes to all measurements in $\XS$. 
Write $\lambdachange{M}{m_k}$
for the assignment obtained from $\lambda$ by changing $M$'s outcome to $m_k$,
    \[\lambdachange{M}{m_k}(B)  
    = \begin{cases} m_k & \text{if $B=M$} \\ \lambda(B) & \text{otherwise.} \end{cases}\]
The condition is then that for any $\lambda$ and any outcome $m_k$ of $M$, one has $b \cdot v_\lambda = b \cdot v_{\lambdachange{M}{m_k}}$.


In terms of the vector of coefficients $b$, this means that, using the equality constraints (viz.\ no-disturbance),
$b$ can be brought into a form that is symmetric under any change of outcome for $M$,
i.e.\ such that  
\begin{equation}\label{eq:symmetric_effective}
b(m_k \, s \mid M \, U) = b(m_j \, s \mid M \, U)
\end{equation}
for all $m_k$, $m_j$ outcomes of $M$,
and all maximal contexts $\{M\} \cup U$ and outcome assignments to the remaining measurements $s \in \OO{U}$. 
With the coefficient vector brought to such a symmetric form, one can unambiguously write $b(s \mid U)$ for the unique coefficient in Eq.~\eqref{eq:symmetric_effective}, without reference to the measurement $M$.
One could, in fact, effectively regard the inequality as being defined over the sub-scenario that excludes the measurement $M$.

We assume throughout that the starting inequality on $\scenS$ has been put in such an \emph{effective} form,
eliminating \emph{all} the measurements not effectively contributing to it.
Note that this means that there is a single $b(s \mid U)$ whenever $U$ is a partial context of $\scenS$ that is maximal within the set of effectively contributing measurements, even if $U$ belongs to more than one maximal context of $\scenS$, when combined with different sets of non-effectively-contributing measurements.

\subsubsection{The form of the lifted inequalities}

Let $b \cdot p\geq 0$ be a facet-defining NC inequality for the scenario $\scenS$.
What we aim, effectively, is to test this inequality in the extended scenario $\scenT$.
This is the premise of lifting:
choosing a vector of coefficients $b'$ in such a way that the inequality $b' \cdot q \geq 0$ over the space of correlations for $\scenT$
remains in essence equivalent to $b \cdot p \geq 0$, but also is a facet-defining NC inequality for $\scenT$.
In other words, the lifted inequality $b' \cdot q \geq 0$ still effectively tests (reduced) statistics on the sub-scenario $\scenS$ but given statistics for $\scenT$.
This requires tweaking with the role of the added measurement $A$ in the expression $b' \cdot q$ by a judicious choice of the coefficients $b'$.

To this end, we must consider two different constructions of coefficient vectors built from $b$, which provide an appropriate answer in different circumstances. In summary (to be expanded below):
\begin{enumerate}[leftmargin=*,align=left,label=\Roman*.,widest*=2]
\item\label{sec: Case 1}
 the coefficient vector $\bItrAZ$
traces $A$ out in the $\CmcA$ components whereas it zeroes out each component in $\CpcA$,
leaving $\Cold$ ones as they appear in the original inequality;

\item the coefficient vector $\bZakak$ enforces only the conditional statistics in $b' \cdot q$:
it fixes an outcome $a_k$ of $A$ within the contexts in $\CmcA \cup \CpcA$ (i.e.~the contexts containing $A$),
treating the corresponding components like the original coefficient vector $b$,
while it zeroes out all components where the outcome of $A$ is not equal to $a_k$, including all components in the $\Cold$ part.
\label{sec: Case 2} 
\end{enumerate}
The subscripts `$\ItrAZ$' and `$\Zakak$' above were chosen to be suggestive of the action of the coefficient vector on each of the three components of the space of correlations as defined in Eq.~\eqref{eq:form-vector-extendedspace} in comparison to that of the initial vector of coefficients $b$:
`$I$' for leaving the action unchanged from $b$, `$0$' for picking coefficients equal to zero, `$\tr(A)$' for tracing out the outcomes of the new measurement $A$, `$a_k$' for selecting on a chosen outcome $a_k$ for $A$.

Depending on the compatibility relations between the newly added measurement $A$ and the preexisting measurements in $\XS$,
one of the above options gives a vector of coefficients for a lifted inequality.
We now explain when to use which form and describe the constructions in more detail.
Bear in mind once again that, in each case, these will only give \emph{one} possible form for the inequality supporting the lifted facet:
since the noncontextual polytopes are not full-dimensional, there is not a unique form (i.e.\ vector of coefficients) for the inequality supporting each facet.

{\em Case I} is used to obtain a lifted inequality when the measurement $A$ is not simultaneously compatible with all of the effectively contributing measurements for the initial inequality $b \cdot p \geq 0$.
In this case, the lifted inequality takes the form
\begin{equation}
 \bItrAZ \cdot q \geq 0. \label{eq:13}
\end{equation}
where we build from $b$ the new vector of coefficients $\bItrAZ$ as follows:
\begin{itemize}
\item over a context $C$ in $\Cold$, we pick the same coefficients as in the initial inequality thus leaving $b$ unchanged:
for each $s \in \OO{C}$,
\begin{equation}\label{eq:bItrAZ_old}
\bItrAZ (s \mid C)  = b(s \mid C);
\end{equation}
\item over a context $\{A\} \sqcup C$ in $\CmcA$, we trace out measurement $A$ by assigning coefficients taken from $b$ irrespective of $A$'s outcome:
for each  $s \in \OO{C}$ and $a \in \OA$,
\begin{equation}\label{eq:bItrAZ_mcA}
\bItrAZ(a \, s \mid A \, C) = b(s \mid C);
\end{equation}
\item a context $\{A\} \sqcup U$ in $\CpcA$ is ignored completely by picking the coefficients to be zero:
for each $s \in \OO{U}$ and $a \in \OA$,
\begin{equation*}
\bItrAZ(a \, s \mid A \, C) = 0.
\end{equation*}
\end{itemize}

To better understand this construction, note that any non-disturbing correlation $q$ on $\scenT$ 
determines a correlation $q|_\scenS$ on the sub-scenario $\scenS$ by marginalisation, forgetting the outcome of $A$ by summing over all the possibilities.
Concretely, given $C \in \CS$, either it is in $\Cold$ and then $q|_\scenS(s \mid C) = q(s \mid C)$,
or it is wholly compatible with $A$ (so that $\{A\} \sqcup C \in \CmcA$)
and then $q|_\scenS(s \mid C) = \sum_{a \in O_A} q(a\, s \mid A \, C)$.
The inequality determined by $\bItrAZ$ effectively tests for the initial inequality $b \cdot p \geq 0$ over such marginal statistics.
That is, for all non-disturbing correlations $q$ on $\scenT$ we have that
\begin{equation}\label{eq:caseImarginal}
\bItrAZ \cdot q \;=\; b \cdot q|_S.
\end{equation}

Later, within the proof of Theorem~\ref{sec:theorem 1}, it will become clear why this construction (case I) is not applicable when $A$ is compatible with all the `effectively contributing measurements' simultaneously.
Despite still yielding a valid inequality, it does not yield a facet-defining one.
For those situations,
we need to consider the other case that we now describe.

{\em Case II} is used to obtain a lifted inequality whenever $A$ is compatible with all the effectively contributing measurements for the initial inequality $b \cdot p \geq 0$.
Recall that we assume that the vector $b$ is in the effective form described in Sec.~\ref{sec:effectivecontribution}.
The explicit form of the lifted inequality is then:  
\begin{equation}
 b_{0,a_k,a_k} \cdot q \geq 0. \label{eq:14}
\end{equation} 
where we build from $b$ the new vector of coefficients $\bZakak$ as follows:
\begin{itemize}
    \item a context $C$ in $\Cold$ is ignored:
    for each $s \in \OO{C}$,
    \begin{equation}\label{eq:bZakak_old}
    \bZakak (s \mid C)  = 0;
    \end{equation}
    \item over a context $\{A\} \sqcup U$ in $\CpcA \cup \CmcA$,
    we pick the coefficient from the initial $b$ when the outcome for $A$ is the fixed $a_k$ and ignore the event otherwise:
    for each  $s \in \OO{U}$ and $a \in \OA$,
    \begin{equation}\label{eq:bZakak_+A}
    \bZakak(a \, s \mid A \, U) = \begin{cases} b(s \mid U) & \text{if $a = a_k$,} \\ 0 & \text{if $a \neq a_k$,}\end{cases}
    \end{equation}
    where $b(s \mid U)$ is well defined due to $b$ being in effective form: see Eq.~\eqref{eq:symmetric_effective} and subsequent explanation, noting that any extension of the context $U$ in $\scenS$ consists solely of measurements incompatible with $A$, which are therefore not effectively contributing to the inequality $b \cdot p \geq 0$.
\end{itemize}

\subsubsection{Recovering the initial inequality}
\label{subsec:recovering}

Since we provide an `if and only if' condition in the measurement lifting theorem (Theorem \ref{sec:theorem 1}),
it is also worth noticing that the constructions of the new (lifted) inequalities over $\scenT$ can be inverted, i.e.\ one can recover the initial inequality given inequalities \eqref{eq:13} and \eqref{eq:14}.
In going back from scenario $\scenT$ to scenario $\scenS$, the contexts in $\Cold$ remain, those in $\CpcA$ are eliminated,
and the measurement $A$ is dropped from each context in $\CmcA$, yielding the set of contexts we denote by $\Cmc$.

Given the coefficient vector $\bItrAZ$, one can easily recover the vector $b$:
for contexts in $\Cold$, the components are unchanged;
for a context $C \in \Cmc$, the components can be read off from the compoenents of $\bItrAZ$ on the corresponding context $C \cup \{A\} \in \CmcA$, since  $b(s \mid C) = \bItrAZ(a \, s \mid A\, C)$ for any $a \in\OA$.

The second case is a little more subtle.
Given the vector $\bZakak$, one may not necessarily recover the original coefficient vector $b$, but one can recover some $b'$ that determines an equivalent inequality.
Note that $\bZakak$ has non-zero coefficients only in the components corresponding to $\CmcA \cup \CpcA$
For each context $U \cup \{A\} \in \CmcA \cup \CpcA$, pick some maximal context $C_U$ in $\CS$ that extends $U$, and let $R_U = C_U \setminus U$ be the set of added measurements. Note that if $U \cup \{A\}$ comes from $\CmcA$, then $U$ is already itself maximal in $\scenS$, being in $\Cmc$, and so $R_U=\varnothing$;
if it comes from $\CpcA$, then $U$ is a partial context of $\scenS$ which may thus be extended to some maximal context in $\Cold$ (in this case the choice of $C_U$ is not necessarily unique).
Now, for any $s \in \OO{U}$ and $r \in \OO{R_U}$, set $b'(s \, r \mid U \, R_U) = \bZakak(a_k \, s \mid A \, U)$.
The components in $b'$ corresponding to contexts not of the form $C_U$ for any $U$ are set to zero.
Note that the assignment $U \mapsto C_U$ is injective: if $U_1, U_2$ are both included in the same maximal context and $U_i \cup \{A\} \in \CpcA$, then $U_1 \cup U_2 \cup \{A\}$ is itself a context including both $U_i \cup \{A\}$,
which forces $U_1=U_2$ by maximality of the contexts in $\CpcA$.


\subsubsection{Validity of the lifted inequalities}
\label{subsec: Validity}

Before proving our main results regarding their facet-defining nature, we prove the validity of these newly defined inequalities for the polytope $T$.


Recall from Sec~\ref{sec:polytope} that it is sufficient to test the inequality on vertices of $T$.
Over each maximal context $C$, a vertex $v$ has precisely one non-zero component, equal to $1$, that of the joint outcome $v|_C \in \OO{C}$. 
Thus, evaluating an inequality  on a vertex  amounts to adding up the coefficients
corresponding to those components, 
a single surviving one for each maximal context.

Moreover, note that any vertex $v$ of $T$ determines a vertex $u$ of $S$ by restricting the global assignment to the measurements in $\XS$.


\begin{proposition}\label{prop:valid-MLcaseI}
Let the scenario $\scenT$ be an extension of the scenario $\scenS$ by the addition of a single measurement $A$.
Then $b \cdot p \geq 0$ is a valid inequality for $S$
if and only if $\bItrAZ \cdot q \geq 0$ is a valid inequality for $T$.
\end{proposition}
\begin{proof}
This essentially amounts to Eq.~\eqref{eq:caseImarginal}, but we now show it explicitly in terms of the vertices.

Let $v$ be a vertex of $T$, and $u = v|_\scenS$ be the corresponding vertex of $S$ obtained by restriction of the global outcome assignment.
The non-zero coefficients in $\bItrAZ$ are all in components relative to contexts in $\Cold \cup \CmcA$.
When evaluating $\bItrAZ \cdot v$, we have:
over a context $C \in \Cold$ the surviving coefficient corresponds to $v|_C = u|_C$, being given as $b(u|_C \mid C)$ by Eq.~\eqref{eq:bItrAZ_old};
over a context $\{A\} \cup C \in \CmcA$, the surviving coefficient corresponds to $v|{\{A\}\cup C}$, being given as
 $\bItrAZ(v|_{\{A\} \cup C} \mid A \, C) = b(u|_C \mid C)$ by Eq.~\eqref{eq:bItrAZ_mcA}.

Since each  context $C \in \CS$ is either in $\Cold$ or extends to one in $\CmcA$, the above cases exhaust all the contexts of $\scenS$ and we obtain that
$\bItrAZ \cdot v \;=\; b \cdot u$, an instance of Eq.~\eqref{eq:caseImarginal}.


 The forward direction of the result then follows from the validity of the initial inequality for $S$, which implies  $b \cdot u \geq 0$.
 Conversely, any vertex $u$ of $S$ can be extended to some vertex $v$ of $T$ by picking any outcome for $A$.
    Validity of the final inequality gives $\bItrAZ \cdot v \geq 0$, hence the same instance of Eq.~\eqref{eq:caseImarginal} implies $b \cdot u \geq 0$.
\end{proof}

\begin{proposition}\label{prop:valid-MLcaseII}
Let the scenario $\scenT$ be an extension of the scenario $\scenS$ by the addition of a single measurement $A$,
and let $b \cdot p \geq 0$ be an inequality over $\scenS$ (written in effective form).
If $A$ is compatible with all the measurements effectively contributing to the inequality $b \cdot p \geq 0$,
then $b \cdot p \geq 0$ is a valid inequality for $S$
if and only if
$\bZakak \cdot q \geq 0$ is a valid inequality for $T$.
\end{proposition}
\begin{proof}
Let $v$ be a vertex of $T$, and $u = v|_\scenS$ be the corresponding restriction to a vertex of $S$.
The only non-zero coefficients in $\bZakak$ correspond to contexts in $\CmcA \cup \CpcA$.

If the vertex $v$ assigns outcome $a_k$ to $A$, then the surviving coefficient at each maximal context $\{A\} \cup U$ is
$\bZakak(v|_{\{A\}\cup U} \mid A \, U) = b(u|_U \mid U)$ by Eq.~\eqref{eq:bZakak_+A}.
We thus obtain $\bZakak \cdot v \;=\; b \cdot u$, and so $\bZakak \cdot v \geq 0$ by validity of the initial inequality for $S$.

If $v$ assigns any other outcome to $A$, then all the surviving coefficients are zero by Eqs.~\eqref{eq:bZakak_old} and~\eqref{eq:bZakak_+A}. Thus, the vertext saturates the inequality, $\bZakak \cdot v = 0$.


Conversely, given a vertex $u$ of $S$, it can be extended to a vertex $v$ of $T$ by assigning outcome $a_k$ to $A$,
and $b \cdot u \geq 0$ follows from $\bZakak \cdot v \;=\; b \cdot u$ and validity of the final inequality.
\end{proof}

\subsubsection{Facetness of the lifted inequalities}


We now move to the core of the measurement lifting result: to show that the constructed inequalities for $T$ are facet-defining. For both cases, the proof follows the general plan outlined  in Sec.~\ref{sec:Proof outline}.

Given the initial inequality $b \cdot p \geq 0$ which supports a facet of $S$, we pick
\[B_S= \{v_1, \ldots, v_{d_S + 1}\}\]
an affine basis for $S$ chosen so that the first 
$d_S$ vertices saturate the inequality.
Any vertex of $S$ extends to a vertex of $T$ by fixing an outcome for the new measurement $A$.
For each outcome $a_k \in \OA$, the corresponding extension of vertex $v_i$ is denoted $v_{i}^{(k)}$.
Thus, the set 
\[V_k=\{v_{1}^{(k)}, \ldots, v_{d_S + 1}^{(k)}\}\]
contains the vertices of $T$ that are extensions of those in $B_S$ by fixing the outcome $a_k$ for $A$.

There is a one-to-one correspondence between vertices of $S$ and vertices of $T$ that have a fixed outcome $a_k$ for $A$.
This correspondence preserves and reflects affine dependencies.
From the fact that $B_S$ is an affine basis for $S$, we can therefore conclude that each set $V_k$ is affinely independent 
and is, moreover, an affine basis for all the vertices of $T$ that assign outcome $a_k$ to $A$.
That is, any vertex that assigns outcome $a_k$ to $A$ is an affine combination of the vertices in $V_k$.
Consequently, the union $B_T = \bigcup_{k} V_k$ affinely spans the polytope $T$.
In other words, it is an overcomplete affine basis for $T$, where `overcomplete' indicates that its elements are not necessarily affinely independent, and so some of them could be discarded with the remaining set still affinely spanning $T$.

\begin{theorem}
\label{sec:theorem 1}
Let the scenario $\scenT$ be an extension of the scenario $\scenS$ by the addition of a single measurement $A$,
and let  $b \cdot p \geq 0$ be an inequality over $\scenS$ (written in effective form). Then,
\begin{enumerate}[leftmargin=*,align=left,label=\Roman*.,widest*=2]
\item if $A$ is incompatible with some measurement that effectively contributes to the inequality, then $b \cdot p \geq 0$ supports a facet of $S$ if and only if $b_{I,tr(A),0} \cdot q \geq 0$ supports a facet of $T$;
\item if $A$ is compatible with all the measurements that effectively contribute to the inequality, then for any outcome $a_k$ of $A$, $b \cdot p \geq 0$ supports a facet of $S$ if and only if $\bZakak \cdot q \geq 0$ supports a facet of $T$.
\end{enumerate}
\end{theorem}

\begin{proof}

We establish each case separately.

\paragraph*{Case I}
First, assume that $b \cdot p \geq 0$ supports a facet of $S$.
We have to show that the inequality $\bItrAZ \cdot q \geq 0$ is saturated by $d_T$ affinely independent vertices and no more. 
Already from the proof of Proposition~\ref{prop:valid-MLcaseI}, or even from Eq.~\eqref{eq:caseImarginal},
we know that $\bItrAZ \cdot v_i^{(k)} = b \cdot v_i$ for all $v_i \in B_S$ and $a_k \in \OA$.
Hence, in each $V_k$, all but the last element saturate the inequality $\bItrAZ \cdot q \geq 0$.

Fixing some $k$, we show next that for all $l \neq k$, the (non-saturating) vertex $v_{d_S+1}^{(l)}$ can be written as an affine combination of $V_l \cup V_k \setminus \{v_{d_S+1}^{(l)}\}$.
This implies that the vertex can be dropped from $B_T$ without affecting its affine span, i.e., the fact that it spans the whole polytope $T$.
Doing so for all $l \neq k$, the only non-saturating point that remains within $B_T$ is $v_{d_S+1}^{(k)}$.
Consequently, among any set of $d_T + 1$ affinely independent points in (the reduced) $B_T$, $d_T$ of them must saturate the inequality $\bItrAZ \cdot q \geq 0$,
implying that the inequality is facet-defining for $T$.

In establishing the above-mentioned affine elimination,
the key idea is to use an affine relation akin to Eq.~\eqref{eq:affinedependence-incompatible} with respect to the measurement $A$ and some effectively contributing measurement incompatible with $A$,
and then expanding on the affine bases $V_k$ and $V_l$.
We now see this in detail.

From the assumption, there is a measurement $M$ that effectively contributes to the initial inequality and is not compatible with $A$. 
Let $w^{(k)}$ be a vertex of $T$ obtained from $v_{d_S+1}^{(k)}$ by changing the outcome of the measurement $M$ to something other than what it is for $v_{d_S+1}^{(k)}$, keeping the rest of the measurement outcomes the same.
Likewise, perform the same change of outcome of $M$ to obtain $w^{(l)}$ from $v_{d_S+1}^{(l)}$.
Considering the four vectors
$\{v_{d_S+1}^{(k)},w^{(k)},v_{d_S+1}^{(l)},w^{(l)}\}$,
we establish that the following equation holds:
\begin{equation}
w^{(k)} - v_{d_S+1}^{(k)}  = w^{(l)} - v_{d_S+1}^{(l)}.
\label{eq:16}
\end{equation}
This has the same form as Eq.~\eqref{eq:affinedependence-incompatible}, and holds for much the same reason, deriving from the incompatibility between the measurements $M$ and $A$.
It follows from the way $w^{(k)}$ and $w^{(l)}$ are constructed from  $v_{d_S+1}^{(k)}$ and $v_{d_S+1}^{(l)}$, respectively:
changing the outcome of $M$ means that no change is made to the components corresponding to contexts in $\CmcA \cup \CpcA$,
since $M$ cannot belong to any such context, incompatible as it is with $A$, i.e.,
\begin{align}
 v_{d_S+1}^{(k)}(\mcA,\pcA) &= w^{(k)}(\mcA,\pcA),  \label{eq:17}
\\
 v_{d_S+1}^{(l)}(\mcA,\pcA) &= w^{(l)}(\mcA,\pcA); \label{eq:18}
\end{align}
moreover,
the fact that $v_{d_S+1}^{(k)}$  and $v_{d_S+1}^{(l)}$ only differ on the outcome they assign to $A$ means that
\begin{equation}
 v_{d_S+1}^{(k)}(\old) = v_{d_S+1}^{(l)}(\old),  \label{eq:19}
\end{equation}
and for the same reason
\begin{equation}
 w^{(k)}(\old) = w^{(l)}(\old). \label{eq:20}
\end{equation}
Eq.~\eqref{eq:16} can then be justified as follows:
the difference between the two terms on the left-hand side cancels out the
components over $\CmcA \cup \CpcA$
due to Eq.~\eqref{eq:17}, and similarly on the right-hand side due to Eq.~\eqref{eq:18};
for the components over $\Cold$, the equation follows by subtracting Eq.~\eqref{eq:19} from Eq.~\eqref{eq:20}. 

Now, due to the symmetry between $V_k$ and $V_l$ and the symmetric construction of $w^{(k)}$ and $w^{(l)}$, we can expand these vectors in the respective affine basis ($V_k$ and $V_l$) with the same coefficients, i.e.,
\begin{align}
 w^{(k)} &= \sum_{i=1}^{d_S+1}\lambda_iv_i^{(k)}, \label{eq:21}
\\
 w^{(l)} &= \sum_{i=1}^{d_S+1}\lambda_iv_i^{(l)}, \label{eq:22}
\end{align}
where $\sum_i \lambda_i = 1$.
This results in
\begin{align}
 w^{(k)} - v_{d_S+1}^{(k)} &= \sum_{i=1}^{d_S}\lambda_iv_i^{(k)} + (\lambda_{d_S+1}-1)v_{d_S+1}^{(k)}
\\
w^{(l)} - v_{d_S+1}^{(l)} &= \sum_{i=1}^{d_S}\lambda_iv_i^{(l)} + (\lambda_{d_S+1}-1)v_{d_S+1}^{(l)}.
\end{align}
Given Eq.~\eqref{eq:16}, we can equate the right-hand sides, obtaining
\begin{equation}
 \sum_{i=1}^{d_S}\lambda_iv_i^{(k)} + (\lambda_{d_S+1}-1)v_{d_S+1}^{(k)} = \sum_{i=1}^{d_S}{\lambda_i}v_i^{(l)} + (\lambda_{d_S+1}-1)v_{d_S+1}^{(l)}.
\end{equation}
Provided $\lambda_{d_S+1} \neq 1$, we can rearrange this as
\begin{equation}
    v_{d_S+1}^{(l)} = v_{d_S+1}^{(k)} - \sum_{i=1}^{d_S}\frac{\lambda_iv_i^{(k)}}{1-\lambda_{d_S+1}}  + \sum_{i=1}^{d_S}\frac{\lambda_iv_i^{(l)}}{1-\lambda_{d_S+1}},
\end{equation} 
exhibitting $v_{d_S+1}^{(l)}$ as an affine combination of $V_k \cup V_l \setminus \{v_{d_S+1}^{(l)}\}$, as desired.

The last step requires that $\lambda_{d_S+1} \neq 1$.  
A value of $\lambda_{d_S+1} = 1$ precludes $v_{d_S+1}^{(l)}$ being expressed as an affine combination of $V_k \cup V_l \setminus \{v_{d_S+1}^{(l)}\}$, thwarting the proof strategy.
One can work around this problem
by picking a different non-saturating vertex of the original inequality as the last element $v_{d_S+1}$ of $B_S$,
or by picking a different outcome of $M$ to which its value changes in going from $v_{d_S+1}^{(k)}$ to $w^{(k)}$.
We defer a careful analysis to Sec.~\ref{anomaly},
where we show that it is always possible to make a choice that avoids the troublesome $\lambda_{d_S+1}=1$
whenever the measurement $M$ effectively contributes to the initial inequality and is incompatible with $A$.


\paragraph*{Case II} We now consider the case when $A$ is compatible with all the effectively contributing measurements to the initial inequality $b \cdot p \geq 0$.
Fix an outcome $a_k$ of $A$ and consider the inequality $\bZakak \cdot q \geq 0$.
In this case, we have that for all $l \neq k$ every vector in $V_l$ saturates the inequality, since as we have seen the left-hand side evaluates to zero.
Only in $V_k$ does the $(d_S + 1)${-th} vertex not saturate the inequality.
The vertex $v_{d_S+1}^{(k)}$ is thus the only element of $B_T$ that does not saturate the inequality.
Since $B_T$ affinely spans the polytope $T$, there must be $d_T + 1$ affinely independent vertices within it.
As only one element of $B_T$ does not saturate the inequality, there is a set of $d_T$ affinely independent vertices that do. Hence, $b_{0,a_k,a_k} \cdot q \geq 0$ represents a facet-defining inequality for $T$.

\paragraph*{Converse}
This completes our proof of measurement lifting in one direction. 
Now we go about proving the converse statements, i.e. that if the inequality \eqref{eq:13} or \eqref{eq:14} is facet-defining for $T$, then the original inequality $b \cdot p\geq 0$ is facet-defining for $S$.
Notice that earlier (in Sec.~\ref{subsec:recovering}) we already showed how to recover $b$
from the given inequalities \eqref{eq:13} or \eqref{eq:14}. 

Given $\bItrAZ \cdot q \geq 0$ or $\bZakak \cdot q \geq 0$ facet-defining,
let $F$ be the set of vertices that saturate it, and partition it into two sets
$F^{(k)}$ and $F^{(\lnot k)}$: those that assign $a_k$ to $A$ and those that do not.
A consequence of the assumed facet-defining nature of the inequality is that
for any non-saturating vertex $v$ (i.e.\ any $v \notin F$), the set $F \cup \{v\}$ affinely spans the whole polytope $T$.
Given the specific inequalities under consideration, one can always choose such a $v$ that assigns outcome $a_k$ to $A$.

Now, let $w$ be any vertex that assigns $a_k$ to $A$.
Expand it as an affine combination of $F \cup \{v\}$:
\begin{equation}
   w = \lambda v +   \sum_{v_i \in F^{(k)}}\lambda_i v_i + \sum_{v_j\in F^{(\lnot k)}}\lambda_j v_j.
\end{equation}
Now, for all vertices that appear in the equation, switch the outcome of $A$ to $a_k$.
This transformation preserves the equality because it is a linear map on the space of correlations:
for contexts in $\CmcA \cup \CpcA$, each component of the form $p(a_k \, s \mid A \, U)$ takes a value given by the linear expression $\sum_j p(a_j \, s \mid A \, U)$ and the remaining components become $0$; for contexts in $\Cold$, the components are unchanged.
It yields
\begin{equation}
   w = \lambda v +   \sum_{v_i \in F^{(k)}}\lambda_i v_i + \sum_{v_j\in F^{(\lnot k)}}\lambda_j v_j[A \mapsto a_k].
   \label{eq:affineFkv}
\end{equation}
Notice that the transformation leaves the vertices in $F^{(k)}$ as well as $v$ and $w$ unchanged,
while those in $F^{(\lnot k)}$ get mapped to vertices in $F^{(k)}$.
Hence, Eq.~\eqref{eq:affineFkv} shows that $w$ can be written as an affine combination of vertices in $F^{(k)} \cup \{v\}$.

Given how $w$ was specified, we have established that $F^{(k)} \cup \{v\}$ affinely spans the set of all vertices of $T$ that assign $a_k$ to $A$.
This set determines a subpolytope of $T$ that is isomorphic to $S$ by the correspondence that forgets the outcome of $A$. 
The vertices of $S$ obtained in this way from $F^{(k)}$ saturate the inequality over $\scenS$ determined by $b$.
We have shown that it is enough to add a single non-saturating vertex, namely $v|_\scenS$, in order to affinely span the whole of $S$. Thus the inequality determines a face of co-dimension $1$, i.e.\ a facet, of $S$.
\end{proof}

\subsubsection{The curious case of {$\lambda_{d_S+1} = 1$}}
\label{anomaly}


We now discharge the issue left not fully resolved in the proof of case I above,
on why one can always avoid
a situation when $\lambda_{d_S+1}=1$.
We do this by examining what such a situation implies about the original inequality being lifted.

Observe that Eqs.~\eqref{eq:21} and~\eqref{eq:22} actually arise at the level of the original scenario $\scenS$ already.
Let $m$ be the outcome of $M$ used in the change from $v_{d_S+1}^{(k)}$ to $w^{(k)}$ in the proof of case I.
That is, adopting the notation from Sec.~\ref{sec:effectivecontribution}, one has $w^{(k)} = v_{d_S+1}^{(k)}[M\mapsto m]$.
At the level of $\scenS$, one may consider a similar change $w = v_{d_S+1}[M\mapsto m]$,
and expanding $w$ as an affine combination of $B_S$ gives
\begin{equation}
w = \sum_{i=1}^{d_S+1}\lambda_i v_i.  
\end{equation}
The case $\lambda_{d_S + 1} = 1$ implies that
\begin{equation}
    w       
    = v_{d_S+1} + \sum_{i=1}^{d_S}{\lambda_i}v_i, 
\end{equation}
where $\sum_{i=1}^{d_S} \lambda_i = 0$.
Under the initial inequality, i.e.\ multiplying the vector of coefficients $b$, this yields
\begin{equation}
    b \cdot w 
   \; = \;
   b \cdot v_{d_S+1} + \sum_{i=1}^{d_S}{\lambda_i}b \cdot v_i
   \; = \;
   b \cdot v_{d_S+1},
\end{equation}
where the last equality follows from the fact that all the vertices $\{v_i\}_{i=1}^{d_S}$ saturate the inequality $b \cdot p \geq 0$.
In other words, the vertex  $w$ 
attains the same value for the inequality as the vertex $v_{d_S+1}$.

Now, note that, in picking the affine basis $B_S$, the last vertex $v_{d_S+1}$ could be chosen arbitrarily to be any non-saturating vertex of the inequality.
Moreover, the outcome $m$ could also be freely chosen.
The fact that the measurement $M$ effectively contributes to the inequality (see Sec~\ref{sec:effectivecontribution}) means that there exists a vertex $v$ and an outcome $m$ of $M$
such that $b \cdot v \neq b \cdot v[M \mapsto m]$.
Without loss of generality, such $v$ can always be chosen to be non-saturating, because one of $v$ or $v[M\mapsto m]$ must not saturate the inequality, as they attain different values.
This guarantees that it is possible to pick a non-saturating vertex $v_{d_S+1}$ and an outcome $m$ of $M$ in such a way that the problematic $\lambda_{d_S+1}=1$ does not occur.

In turn, if the measurement $M$ does not effectively contribute to the inequality,
then $\lambda_{d_S+1} = 1$ for all choices of $v_{d_S+1}$ and outcome $m$.
We illustrate such a situation through a (paradigmatic) example where the newly added measurement is compatible with all the measurements that effectively contribute to the initial inequality,
and check that the case I construction indeed fails to lift it to a facet-defining inequality.

Consider the scenario given by the compatibility graph shown in Fig.~\ref{fig: fig2}(a).
The following is a facet-defining NC inequality to which measurement `4' does not effectively contribute:
\begin{multline}
p(010|014)+ p(011|014) + p(100|014) + p(101|014) \\
+ p(01|12) + p(10|12) + p(01|23) + p(10|23)\\
- p(01|30) - p(10|30)  \geq 0.  \label{eq:31}
\end{multline}
In fact, this inequality is a lifting of the CHSH inequality (with the classical bound set to zero) over the induced sub-scenario with measurement set $\{0,1,2,3\}$ through the addition of measurement `4'.
Notice how the measurement `4' is being `traced out', and so the inequality is still, in essence, the CHSH inequality:
\begin{multline}
\ICHSH = p(01|01) + p(10|01)  + p(01|12) + p(10|12) \\ +   p(01|23) + p(10|23)  - 
p(01|30) - p(10|30) \geq 0.  \label{eq:32}
\end{multline}
Now, consider the extension of the scenario in Fig.~\ref{fig: fig2}(a) to the scenario in Fig.~\ref{fig: fig2}(b).
The newly added measurement `5' is compatible with all the measurements that effectively contribute to the inequality \eqref{eq:31}.
This indicates that, in order to lift the inequality, one must use case II.
Indeed, doing so gives rise to two facet-defining inequalities, one for each outcome $a_k \in \OO{5} = \{0,1\}$:
\begin{multline*}
I_{a_k} = p(10a_k|015) + p(01a_k|015) + 
  p(10a_k|125) + \\ p(01a_k|125) + p(10a_k|235) + p(01a_k|235) \\ - p(01a_k|305) - p(10a_k|305) \geq 0.
\end{multline*}
Had we tried to apply case I, which would amount to tracing measurement `5' out, we would have obtained the inequality
\begin{equation*}
    \ICHSH = \sum_k I_{a_k} \geq 0,
\end{equation*}
which is still, in essence, the CHSH inequality on the induced sub-scenario with measurement set $\{0,1,2,3\}$.
This is a valid but not a facet-defining NC inequality for the scenario in Fig.~\ref{fig: fig2}(b).

The case just described exemplifies what happens in general:
in situations that fall within the scope of case II, the construction of case I leads to a valid but non-facet-defining inequality, which can be written as a linear combination of the lifted inequalities obtained through case II.

\begin{figure}
    \centering
    \includegraphics[scale = 0.45]{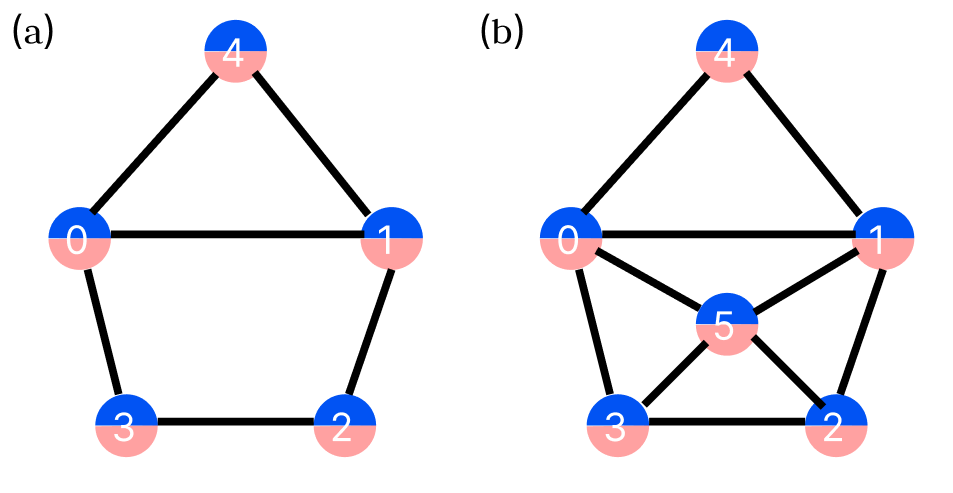}
    \caption{(a) Compatibility graph of the initial scenario with the node labelled `4' representing a measurement added to the CHSH scenario. Bi-coloured nodes represent dichotomic measurements and edges represents compatibility between measurements. (b) Compatibility graph of scenario obtained from (a) by adding a new measurement (node labelled `5') compatible with all the measurements effectively contributing to inequality \eqref{eq:31}.} 
\label{fig: fig2}
\end{figure}

\subsection{Outcome lifting}
\label{sec:outcome lifting}



We now move on to study lifting under outcome extensions.
Let $\scenS$ be extended to $\scenT$ by the addition of a new outcome $a_0$ for the existing measurement $A \in \XS$.
That is, in $\scenS$ the measurement $A$ has outcome set $\OSA=\{a_1, \ldots, a_m\}$,
which is extended in $\scenT$ to $\OTA=\{a_0\} \cup \OSA = \{a_0, a_1, \ldots, a_m\}$.

The maximal contexts of $\scenT$ are the same as those of $\scenS$.
The ambient vector space of correlations for $\scenT$ has components matching those for $\scenS$ and additional ones referring to the new outcome $a_0$. Namely, for each maximal context containing $A$, $\{A\} \cup U \in \CS$, and each joint outcome for the rest of the measurements, $s \in \OO{U}$, 
a vector $q$ on the space of correlations for $\scenT$ has an additional component $q(a_0 \, s \mid A \, U)$.

To exhibit lifting, we must first map an inequality over $\scenS$ into one over $\scenT$.
For this, we `club together' the outcome $a_0$ with some other pre-existing outcome $a_k$ in $\OSA$,
treating the new outcome as if it were $a_k$.
More explicitly, let 
\begin{equation}
 b \cdot p \geq 0 \label{eq:35}
\end{equation}
be the initial facet-defining inequality for $\scenS$.
By `clubbing together' with $a_k$ we mean that wherever a  probability 
$p(a_k \, s \mid A \, U)$ (for some $\{A\} \cup U\in \CS$ and $s \in \OO{C}$)
appears in the inequality \eqref{eq:35} we replace it by $q(a_0 \, s \mid A \, U) + q(a_k \, s \mid A \, U)$ in the new inequality over $\scenT$.

Let us carefully define the new inequality by giving its vector of coefficients.
From the vector of coefficients $b$ of the original ineqality over $\scenS$,
we construct two new vectors over the vector space of correlations for $\scenT$:
\begin{itemize}
    \item the vector $b_*$ is equal to the original vector $b$ padded with zeros on the new components, those that correspond to observing outcome $a_0$ for measurement $A$ (more precisely, to joint outcome assignments for maximal contexts containing the measurement $A$ which assign outcome $a_0$ to $A$).
    \item the vector $b_{a_k}$ is non-zero only on the new components (events assigning outcome $a_0$ to $A$),
    taking the same coefficient as $b$ does
    on the event obtained by changing $A$'s outcome to $a_k$:
    that is, for each $\{A\} \cup U$ a maximal context and $s \in \OO{U}$,
    \[b_{a_k}(a_0 \, s \mid A \, U) \;=\; b(a_k \, s \mid A \, U),\]
    and $b_{a_k}$ is zero at every other component.
\end{itemize}
The new vector of coefficients 
is $b_* + b_{a_k}$, yielding the following inequality over the space of correlations for $\scenT$:
\begin{equation}
 b_* \cdot q + b_{a_k} \cdot q \geq 0. \label{eq:36}
\end{equation}

Much like a case I measurement-lifted inequality could be understood through the marginalisation map that takes each correlation $q$ on $\scenT$ to the correlation $q|_\scenS$ on $\scenS$,
so one way to understand this inequality is through a \emph{coarse-graining} map from correlations on $\scenT$ to correlations on $\scenS$.
Given a correlation $q$ on $\scenT$, define a correlation $q_{k \equiv 0}$ on $\scenS$ by `clubbing together' outcomes $a_0$ and $a_k$ for $A$, as follows:
\begin{align*}
    q_{k \equiv 0}(a_k \, s \mid A \, U) = q(a_k \, s \mid A \, U) + q(a_0 \, s \mid A \, U)
\end{align*}
for $\{A\} \cup U \in \CS$ and $s \in \OO{U}$,
and $q_{k \equiv 0}(s \mid C) = q(s \mid C)$ in all other cases.
Then, the lifted inequality \eqref{eq:36} effectively tests for the original inequality \eqref{eq:35}
on the coarse-grained statistics.
That is, analogously to Eq.~\eqref{eq:caseImarginal}: for all correlations $q$ on $\scenT$,
\begin{equation*}
b_* \cdot q + b_{a_k} \cdot q \;=\; b \cdot q_{k \equiv 0} .
\end{equation*}
The mapping $q \mapsto q_{k \equiv 0}$ clearly preserves deterministic non-contextual correlations,
as it can be seen to work at the level of global outcome assignments.
These correlations are the vertices of the noncontextual polytopes.
The mapping is therefore a map of polytopes $T \to S$.
This observation is the key to prove the validity of the lifted inequality \eqref{eq:36},
which we now explicitly state (cf.\ Propositions~\ref{prop:valid-MLcaseI} and~\ref{prop:valid-MLcaseII}, the analogous statements for measurement lifting).

\begin{proposition}
Let the scenario $\scenT$ be an extension of the scenario $\scenS$ by the addition of a single outcome $a_0$ to an existing measurement $A$.
If $b \cdot p \geq 0$ is a valid inequality for $S$,
then $(b_*+b_{a_k}) \cdot q \geq 0$ is a valid inequality for $T$.
\end{proposition}
\begin{proof}
    Let $v$ be any vertex of the polytope $T$, corresponding to an global outcome assigment $\lambda \in \OX$.
    If it assigns to $A$ an outcome other than $a_0$, 
    then the second term in $b_* \cdot v + b_{a_k} \cdot v$ vanishes
    and the first reduces to $b \cdot v_{k \equiv 0}$ where $v_{k \equiv 0}$ is the vertex of $S$ with precisely the same underlying global outcome assignment $\lambda$ as $v$.
    If it assigns outcome $a_0$ to $A$, then the first term in $b_* \cdot v + b_{a_k} \cdot v$ vanishes and the second reduces to $b \cdot v_{k \equiv 0}$,
    where now  $v_{k \equiv 0}$ is the vertex of $S$ whose underlying global outcome assignment is $\lambda[A \mapsto a_k]$, i.e. it assigns outcome $a_k$ to $A$ while retaining the same outcomes as $\lambda$ for the other measurements.
    The result then follows from the validity of $b \cdot p \geq 0$.
\end{proof}

While the inequality \eqref{eq:36} is always valid for $T$, it is not always a lifting of the initial inequality \eqref{eq:35}, as it may fail to be facet-defining.
We now describe a (general) situation in which it fails to lift.

Suppose that the initial inequality \eqref{eq:35} is such that all non-zero coefficients in $b$ correspond to joint outcomes that assign the fixed outcome $a_k$ for $A$, i.e.~only components of the form $b(a_k \, s \mid A \, U)$ may be non-zero.
This renders it an inequality of the form \eqref{eq:14}, that is, $b = c_{0,a_k,a_k}$ for some vector of coefficients $c$ over the sub-scenario $\scenR$ obtained from $\scenS$ by dropping the measurement $A$.
So, the initial inequality over $\scenS$ is itself a lifting from the smaller scenario  $\scenR$ via case II measurement lifting.
If one then chooses to club $a_0$ with $a_k$,
one ends up getting that
$b_* = c_{0,a_k,a_k}$
(where now the right-hand side represents a case II measurement lifting of the inequality determined by $c$ directly from $\scenR$ to $\scenT$),
and
$b_{a_k} = c_{0,a_0,a_0}$
(where the right-hand side also represents a case II measurement lifting directly from $\scenR$ to $\scenT$).
Consequently, the candidate lifted inequality \eqref{eq:36} would read
\begin{equation*}
    c_{0,a_k,a_k} \cdot q + c_{0,a_0,a_0} \cdot q \geq 0.
\end{equation*}
This is clearly not a facet-defining inequality for $T$, since by Theorem~\ref{sec:theorem 1} each of the two summands is itself a facet-defining inequality,
obtained via case II measurement lifting from $\scenR$ to $\scenT$.
In our outcome lifting result, we must therefore exclude such cases.
Note that such an inequality determined by $b = c_{0,a_k,a_k}$ (case II measurement lifted from $\scenR$ to $\scenS$)
can be lifted to $\scenT$ if we choose to club $a_0$ with a different outcome $a_j \neq a_k$.
In that case, however, the vector $b_{a_j}$ is equal to zero and the lifted inequality is simply
$c_{0,a_k,a_k} \cdot q \geq 0$, which can also be obtained via case II measurement lifting directly from $\scenR$ to $\scenT$.

With the proviso of excluding the situation discussed in the previous paragraph,
we now prove the facet-defining nature of the noncontextuality inequality \eqref{eq:35} in all other situations
(cf.\ Theorem~\ref{sec:theorem 2} for measurement lifting).

\begin{theorem}
\label{sec:theorem 2}
Let the scenario $\scenT$ be an extension of the scenario $\scenS$ by the addition of a single outcome $a_0$ to an existing measurement $A$.
If $b \cdot p \geq 0$ is a facet-defining NC inequality for $\scenS$, then for any pre-existing outcome $a_k$ of $A$ for which $b \neq c_{0,a_k,a_k}$, the inequality $(b_*  + b_{a_k}) \cdot q \geq 0$ is a facet-defining NC inequality for $\scenT$.
\end{theorem}
\begin{remark*}
    The condition that $b$ not be of the form $c_{0,a_k,a_k}$, i.e.~not be obtainable by case II measurement lifting when adding measurement $A$ with choice of outcome $a_k$, admits an equivalent formulation in terms of the properties of the inequality being lifted (or of the facet it defines).
    The case being excluded is when the measurement $A$ is compatible with all the measurements effectively contributing to the inequality $b \cdot p \geq 0$ \textit{and}, moreover, 
    all the vertices of $S$ that do not saturate the inequality assign outcome $a_k$ to $A$ (i.e. $b \cdot v = 0$ for all $v$ not assigning $a_k$ to $A$).
\end{remark*}
\begin{proof}
As in the proof for measurement lifting (Theorem~\ref{sec:theorem 2}),
the strategy is to use an affine basis for $S$ to define a set of vertices spanning $T$, from which one can affinely eliminate all but one vertex not saturating the lifted inequality.
We start by picking a set of affinely independent vertices of $S$,
\[B_S = \{v_1, \ldots, v_{d_s+1}\},\]
in such a way that the first $d_S$ vectors saturate the initial facet-defining inequality $b \cdot p \geq 0$.
Recall that the last vertex, $v_{d_s+1}$, may be chosen arbitrarily from among the vertices of $S$ not saturating the inequality.

We now form two sets of vertices of $T$ from $B_S$:
\begin{itemize}
    \item $V^*=\{v^*_1, \ldots, v^*_{d_S+1}\}$, where each $v^*_i$ is the vertex of $T$ corresponding to the same underlying global outcome assignment as $v_i$
    (n.b.~even though $v_i$ and $v^*_i$ have the same underlying outcome assignment, they are different vectors living in vector spaces with different dimensions: $v^*_i$ is obtained from $v_i$ by padding it with zeroes in the new components);
       \item $V' = \{v'_1, \ldots, v'_{d_S+1}\}$, where each $v'_{i}$ is the vertex of $T$ whose underlying assignment is obtained from that of $v_i$ by substituting the outcome $a_0$ for the outcome $a_k$ if the latter is the outcome that $v_i$ assigns to $A$
\end{itemize}
Observe that if $v_i$ assigns an outcome other than $a_k$ to $A$, then $v'_i = v^*_i$.

The union of these two sets
$B_T = V^* \cup V'$
affinely spans the whole polytope $T$, i.e.\ it is a (potentially over-complete) affine basis for $T$.
Both $v^*_i$ and $v'_i$
are mapped to $v_i$ under the coarse-graining $T \to S$, i.e.\
$(v^*_i)_{k \equiv 0} = (v_i')_{k\equiv 0} = v_i$.
Consequently, the last vertex in each of $V^*$ and $V'$ 
is the only one that does not saturate the inequality $(b_*+b_{a_k})\cdot q \geq 0$.

We divide the proof into two cases:
\begin{itemize}
 \item the non-saturating vertex $v_{d_S+1}$ in $B_{S}$ can be picked so that it does not assign outcome $a_k$ to $A$. 
 \item every vertex of $S$ not saturating the initial inequality assigns outcome $a_k$ to $A$. 
\end{itemize} 

The first case is quickly discharged.
We have $v^*_{d_S+1}=v'_{d_S+1}$, i.e.~the last vertex within $V^*$ and the last vertex within $V'$ are  exactly the same.
Hence, $B_T$ contains only one vertex not saturating the inequality.
It must therefore containt $d_{T}$ affinely independent points that saturate it, proving that it defines a facet of $T$.

In the second case, the last vertices of $V^*$ and $V'$ are different, differing only on the outcome each assigns to the measurement $A$, respectively $a_k$ and $a_0$.
We show that $v'_{d_S+1}$, the last element of $V'$, can be eliminated as an affine combination of the remaining vertices in $B_T$, using a similar approach as for measurement lifting.

By the remark below the theorem statement,
in this case, there exists a measurement $M$ which is incompatible with $A$ and which effectively contributes to the initial inequality.
Let $w$ be a vertex of $S$ obtained from $v_{d_S+1}$ by changing the outcome of $M$ to a different value.
It can be expanded in the affine basis $B_S$ as 
\begin{equation}
w \;=\; \sum_{i=1}^{d_S+1}\lambda_{i}v_{i},
\label{eq:OL:affineexpansion_w}
\end{equation}
with $\sum_{i=1}^{d_S+1}\lambda_{i} = 1$.
For the same reason discussed in Sec.~\ref{anomaly}, the vertex $w$ (more precisely, the outcome it assigns to $M$) may be chosen so that $\lambda_{d_S+1} \neq 1$.
In that case,
we would have
\begin{equation}
w \;=\; v_{d_S+1} + \sum_{i=1}^{d_S} \lambda_{i}v_{i},
\end{equation}
with $\sum_{i=1}^{d_S}\lambda_i = 0$, and applying $b$,
\begin{equation}
b \cdot w \;=\; b \cdot v_{d_S+1} + \sum_{i=1}^{d_S} b \cdot \lambda_{i}v_{i} \;=\; b \cdot v_{d_S+1},
\end{equation}
meaning that $b$ does not distinguish between the different outcomes that $v_{d_S+1}$ and $w$ assign to $M$.
This can be avoided because $M$ effectively contributes to the inequality.

Now, let $w^*$ and $w'$ be the vertices of $T$ built from $w$ like $v^*_i$ and $v'_i$ are built from $v_i$,
respectively by considering the same global outcome assignment and the one where $A$'s outcome is changed to $a_0$.
Analogously to Eq.~\eqref{eq:16} in the measurement lifting proof, we have the following equality
\begin{equation}
 w^{*}-v^{*}_{d_S+1}  = w' - v'_{d_S+1} . \label{eq:43}
\end{equation}
To see this explicitly,
divide the maximal contexts in $\CT$ into two mutually exclusive subsets -- those without $A$ and those with $A$ --
and consider the corresponding direct sum decomposition of the vector space into two summands.
The vectors $v^{*}_{d_S+1}$ and $v'_{d_S+1}$ (similarly, $w^{*}$ and $w'$) only differ in the second summand, as their underlying global assignments differ only for $A$.
On the other hand, since $M$ is incompatible with $A$, it only appears in maximal contexts not containing $A$,
and so $v^*_{d_S+1}$ and $w^*$ (similarly,  $v'_{d_S+1}$ and $w'$) only differ in the first summand.
In summary, we can write these four vectors as
\begin{align*}
 v^*_{d_S+1} &= [u_{m_0},u_{a_k}]
 &
 w^* &= [u_{m_1},u_{a_k}]
 \\
 v'_{d_S+1} &= [u_{m_0},u_{a_0}]
 &
 w' &= [u_{m_1},u_{a_0}]
\end{align*}
for some $u_{m_0}, u_{m_1}$ in the first summand and $u_{a_0}, u_{a_k}$ in the second, where the indices are suggestive of the differing outcomes for measurements $M$ and $A$. Eq.~\eqref{eq:43} is now evident from inspection.

Again as before, the vertices $w^*$ and $w'$ are affine combinations of those in $V^*$ and $V'$, respectively, with the same coefficients as in Eq.~\eqref{eq:OL:affineexpansion_w}:
\begin{align*}
 w^* &= \sum_{i=1}^{d_S+1}\lambda_{i}v^*_{i}, 
 &
 w' &= \sum_{i=1}^{d_S+1}\lambda_{i}v'_{i},
\end{align*}
Putting these expansions into Eq.~\eqref{eq:43} yields 
\begin{equation*}
     \sum_{i=1}^{d_S}\lambda_i v^*_{i} + (\lambda_{d_S+1}-1)v^{*}_{d_S+1} = \sum_{i=1}^{d_S}\lambda_i v'_{i} + (\lambda_{d_S+1}-1)v'_{d_S+1},
\end{equation*}
and given that $\lambda_{d_S+1} \neq 1$, 
\begin{equation*}
 v'_{d_S+1} = v^*_{d_S+1} - \sum_{i=1}^{d_S}\frac{\lambda_i v^{*}_{i}}{1-\lambda_{d_S+1}} + \sum_{i=1}^{d_S}\frac{\lambda_i v'_{i}}{1-\lambda_{d_S+1}}.
\end{equation*}
This exhibits $v'_{d_S+1}$ as an affine combination of the remaining elements of $B_T$,
showing it is redundant in $B_T$.
In other words, the set \ $B_T \setminus \{ v'_{d_S+1}\}$ still affinely spans $T$.
But it now contains only one vertex that does not saturate the inequality 
$(b_*  + b_{a_k}) \cdot q \geq 0$.
Necessarily, then, among $d_T+1$ affinely independent vertices in this set, $d_T$ of them must saturate the inequality, showing that it supports a facet of $T$.
\end{proof}

\section{Applications}
\label{sec:applications}

To demonstrate its power, we apply our lifting technique to extract facet-defining noncontextuality inequalities of scenarios
for which no such facet inequalities had been hitherto described.
To do so, we pick two well-known inequalities, the Clauser--Horne--Shimony--Holt (CHSH) and the Klyachko--Can--Binicio\u{g}lu--Shumovsky (KCBS) inequalities, as our starting points.
In Sec.~\ref{sec:App1}, we sequentially measurement lift the former to extract a facet-defining NC inequality for the scenario described by the anti-heptagon $\Bar{C}_7$, the complement of the $7$-cycle, with binary outcomes.
In Sec.~\ref{sec:App2}, we sequentially outcome lift the latter to obtain a facet-defining NC inequality for the 5-cycle scenario with three outcomes per measurement.


\subsection{Lifting the CHSH inequality to the \\two-outcome anti-heptagon scenario}
\label{sec:App1}


\begin{figure}
    \centering
    \includegraphics[scale = 0.70]{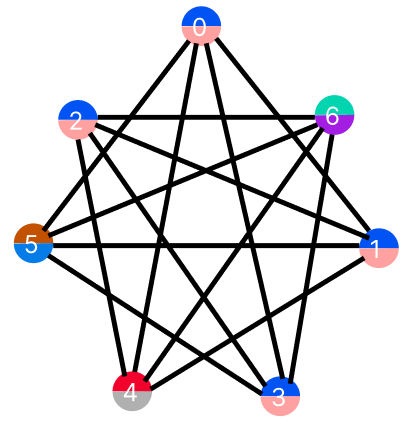}
    \caption{Comaptibility graph of the antiheptagon ($\Bar{C}_7$) scenario with bi-coloured nodes representing dichotomic measurements.}
    \label{fig: fig3}
\end{figure}


In our first example, we apply sequential measurement lifting to the CHSH inequality, defined on the $4$-cycle scenario,
and derive a facet-defining inequality for the scenario with compatibility graph $\Bar{C}_7$, the complement of the $7$-cycle, shown in Fig.~\ref{fig: fig3}.
All measurements are dichotomic in this example, taking values in the outcome set $\{0,1\}$.

We start from the sub-scenario given by the induced subgraph of $\Bar{C}_7$ with node set $\{0,1,2,3\}$ as per Fig.~\ref{fig: fig3}, a $4$-cycle.
The initial facet-defining inequality (CHSH) is 
\begin{multline}
p(00|01)+ p(11|01) + p(00|12) 
+ p(11|12)\\+ p(00|23)+ p(11|23)  
+ p(01|30)+ p(10|30) \leq 3. \label{eq:50}
\end{multline}
We now add measurement `4' to this sub-scenario, with the compatibilities induced from $\Bar{C}_7$.
This leads to a new scenario with maximal contexts $\{0,1,4\}$, $\{1,2,4\}$, $\{2,3\}$, $\{0,3\}$.
The way in which the new measurement is affixed to the initial scenario, namely the fact that it is not compatible with measurement `3', means that case I of measurement lifting applies.
Hence, we simply need to trace out the added measurement.
Inequality \eqref{eq:50} then turns out to be facet-defining for this scenario as well.
Note that the term $p(00|01)$ can be written as $p(000|014)+p(001|014)$ in terms of probabilities over maximal contexts of the extended scenario.
We then similarly add measurement `5'.
This yields a new scenario with maximal contexts $\{0,1,4\}$, $\{1,2,4\}$, $\{2,3\}$, $\{0,3,5\}$, $\{0,1,5\}$.
Again, case I of measurement lifting applies since the added `5' is incompatible with measurement `2',
and the inequality \eqref{eq:50} remains facet-defining for this scenario.
Finally, to get to the anti-heptagon, we add measurement `6', arriving at the scenario with compatibility graph $\Bar{C}_7$.
Its maximal contexts are $\{0,1,4\}$, $\{1,2,4\}$, $\{0,3,5\}$, $\{0,1,5\}$, $\{2,3,6\}$, $\{2,4,6\}$, $\{3,5,6\}$.
Once more, case I applies since the added `6' is incompatible with measurements `0' and `1',
and therefore inequality \eqref{eq:50}
is facet-defining for the anti-heptagon scenario.
Note that the expression in \eqref{eq:50} can of course be expanded in terms of probabilities over maximal contexts of the anti-heptagon scenario.
It is the first-ever identified facet-defining NC inequality for this scenario.  

Note that it is relatively straightforward to see that \eqref{eq:50} must be a valid inequality for the scenario $\Bar{C}_7$, as for any other scenario that contains the initial $4$-cycle as a sub-scenario (even if not an induced one).
The crucial point is that, in this case, it supports a facet and not a lower-dimensional face.


\subsection{Lifting the KCBS inequality to the three-outcome five-cycle scenario}
\label{sec:App2}


\begin{figure}
    \centering
    \includegraphics[scale = 0.70]{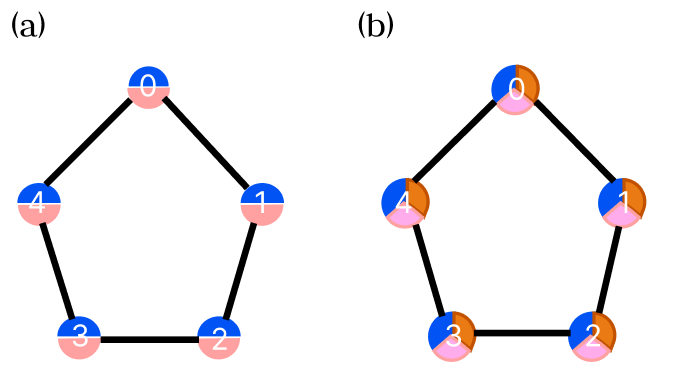}
    \caption{(a) Compatibility graph of the initial KCBS $5$-cycle scenario with five dichotomic measurements, represented by bi-coloured nodes.
    (b) Compatibility graph of the $5$-cycle scenario with trichotomic measurements, represented by tri-coloured nodes. Inequality \eqref{eq:57} is a facet-defining NC inequality for this extended scenario, obtained via sequential outcome lifting of the KCBS inequality.}
    \label{fig: fig4}
\end{figure}


In the second example, we apply sequential outcome lifting to the KCBS inequality, defined on the two-outcome $5$-cycle scenario,
and obtain a facet-defining NC inequality for the scenario with the same compatibility graph, $C_5$, but now with each measurement having three allowed outcomes;
see Fig.~\ref{fig: fig4}.

The starting facet-defining inequality (KCBS) is
\begin{equation}
p(01|01) + p(01|12) + p(01|23) + p(01|34) + p(01|40) \leq 2. \label{eq:54}
\end{equation}
We first add a new outcome, labelled 2, to measurement `0'.
We choose to club this new outcome with the pre-existing outcome labelled 1.
Then according to our lifting method, in the inequality \eqref{eq:54}, where outcome 1 for measurement `0' appears in the context $\{0,4\}$, the term $p(01|40)$ has to be replaced by $p(01|40) + p(02|40)$ to obtain the new facet-defining inequality:
\begin{multline}
p(01|01) + p(01|12) + p(01|23)  \\
+ p(01|34) + p(01|40) + \bm{p(02|40)} \leq 2. \label{eq:55}
\end{multline}
Likewise, we add a new outcome, labelled 2, to measurement `1', and club it with its old outcome labelled 1.
Tweaking the term $p(01|01)$ in the inequality \eqref{eq:55} as before, we obtain a new facet-defining inequality:
\begin{multline*}
p(01|01) + \bm{p(02|01)} + p(01|12) + p(01|23)  \\
+ p(01|34) + p(01|40) + p(02|40) \leq 2.
\end{multline*}
We continue like this by adding a new outcome, labelled 2, to the remaining measurements `2', `3', and `4', and clubbing it with the respective outcome labelled 1.
In the end, we arrive at the following facet-defining inequality for the trichotomic $5$-cycle scenario:
\begin{multline}
p(01|01) + \bm{p(02|01)} + p(01|12) + \bm{p(02|12)}  \\
+ p(01|23) + \bm{p(02|23)} +  p(01|34) + \bm{p(02|34)} \\ 
+ p(01|40) + \bm{p(02|40)} \leq 2, \label{eq:57}
\end{multline}
where the new terms not present in the initial inequality are highlighted.

There was an element of choice in the way the above lifting of the KCBS inequality was obtained.
We could just as well have chosen to club the new outcome with the outcome labelled 0 rather than that labelled 1. Indeed, we could have made a different choice for each measurement. Each such sequence of choices determines a different facet-defining inequality for the final scenario, each corresponding to a different coarse-graining of outcomes. Thus, our method yields not one but several new facet-defining inequalities for the three-outcome $5$-cycle scenario as sequential outcome liftings of the KCBS inequality.

\subsection{Further remarks}
\label{sec:AppDiscussion}

We conclude this section with some more general remarks about applying our lifting technique.

Any contextuality-witnessing scenario must contain an $n$-cycle (with $n \geq 4$) as an induced sub-scenario. This follows from Vorob{\textquotesingle}ev's theorem \cite{Vorobyev:1962}.
Combining this fact with the results of Ref.~\cite{PhysRevA.88.022118}, where the authors completely characterise the facets for dichotomic cycle scenarios, leads us to conclude that we can provide via lifting many facet-defining NC inequalities for any scenario that can witness contextuality.
Notice that, in the two examples detailed above, we did in fact lift two cycle inequalities, respectively CHSH and KCBS.

We also highlight an observation about applying our technique
to sequentially lift an inequality from an initial scenario $\scenS$ to a final scenario $\scenT$.
When measurement lifting through the addition of two measurements, say $A$ and $B$,
the lifted inequality might depend on whether one first adds $A$ and then $B$ or vice-versa.
This means that, in order to obtain all the facet-defining inequalities of $T$ that are lifted from $S$, one must consider all the possible orderings to reach scenario $\scenT$ by extending $\scenS$.


\section{Outlook}
\label{sec:discussion}


We have extended the method introduced by Pironio for lifting facet-defining Bell inequalities to arbitrary contextuality scenarios described by compatibility graphs.
The method allows us to produce facet-defining noncontextuality inequalities in all scenarios that admit contextual correlations.
For most of these scenarios, no such inequalities had been previously described.

Our work invites investigations in the spirit of Ref.~\cite{Pironio_2014} to identify new contextuality scenarios whose noncontextual polytope can be fully characterised via liftings, e.g.\ where every facet-defining NC inequality is obtained via lifting from some particularly simple class of noncontextuality inequalities such as cycle inequalities \cite{PhysRevA.88.022118}.
Complementarily, one may focus on identifying facet-defining NC inequalities that cannot be obtained via lifting from any sub-scenario.
As such, our work can be seen as a step towards complete characterisations of noncontextual polytopes beyond the few known scenarios \cite{PhysRevA.88.022118}.

Another aspect that is worth clarifying, in the wake of the comments in Sec.~\ref{sec:AppDiscussion}, is how the various forms of lifting compose. This would help avoid redundancy, trimming the search space for inequalities lifted from sub-scenarios.

Future applications of our method may include finding facet-defining inequalities in different kinds of scenarios such as Bell scenarios with sequential compatible measurements \cite{Cabello:2010PRL,Liu:2016PRL,PhysRevResearch.5.L012035}, extended Bell scenarios \cite{Extended_Bell:2023_PRL}, or scenarios incorporating causality \cite{fritz2016beyond2,gogioso2021sheaf,abramsky2024combining}.

Similarly, this work dealt only with scenarios where measurement compatibility is described by a binary relation, conforming with Specker's principle \cite{Cabello:2012XXX} in quantum mechanics; that is, a set of measurements is compatible (and thus forms a context) if it is pairwise so. However, the general theory of contextuality, as in e.g.\ Refs.~\cite{Abramsky:2011NJP,AcinCMP2015}, admits more general forms of compatibility, described by simplicial complexes (or hypergraphs).
One might ask how lifting works in that more general setup.

Moreover, it seems that the liftings from $\scenS$ to $\scenT$ considered in this work may be understood through an underlying classical procedure $\scenT \to \scenS$, a free operation in the resource theory of contextuality \cite{Abramsky:2017PRL,AmaralPRL2018,Barbosa2023Closing}
which maps correlations on $\scenT$ to correlations on $\scenS$ in a contextuality non-increasing fashion.
For outcome lifting, this is a coarse-graining operation; for case-I measurement lifting, it is restriction or marginalisation; for case-II measurement lifting, it appears to be a sub-normalised operation performing a form of post-selection.
There is scope to further explore and clarify this resource-theoretic perspective in future.


We conclude on a more practical note,
outlining some relevant potential applications for our lifting technique.
As mentioned, contextuality has been linked to quantum advantage in a number of information-processing tasks.
One prominent example is that of communication complexity protocols: it is shown in Ref.~\cite{gupta2023quantum} that a violation of a NC inequality witnesses quantum advantage in a suitably designed protocol.
New NC inequalities obtained by lifting may therefore capture quantum advantage in new, suitably related communication complexity protocols.
Similar analyses can be carried out
for other information-processing tasks whose success probability is given by a linear functional that determines a NC inequality.

Another aspect that can be explored for newly derived NC inequalities is self testing.
An inequality is self testing when a maximal quantum violation certifies a unique quantum realization (up to global isometries).
It has been studied both for state-dependent \cite{bharti2019robust} and state-independent \cite{xu2023state} inequalities.
An open question is whether the self-testing property is preserved under liftings.
Furthermore, Ref.~\cite{saha2019state} showed that any state-independent violation of a NC inequality offers quantum advantage in a specific, suitably designed distributed computing task. 
Deriving new such inequalities using our lifting method might therefore provide new distributed computing protocols exhibiting quantum advantage.


\section*{Acknowledgments}
We thank Som Kanjilal, Rafael Wagner, and Martti Karvonen for valuable comments and discussions.

The work of all three authors was supported by the Digital Horizon Europe project \href{https://doi.org/10.3030/101070558}{FoQaCiA}, \textit{Foundations of Quantum Computational Advantage}, GA no. 101070558.
R.C.\ and R.S.B.\ also acknowledge financial support from FCT -- Funda\c{c}\~ao para a Ci\^encia e a Tecnologia (Portugal) through PhD Grant SFRH/BD/151452/2021 (R.C.) and through CEECINST/00062/2018 (R.S.B.).
A.C.\ was also supported by the \href{https://doi.org/10.13039/501100011033}{MCINN/AEI} (Spain), Project No.\ PID2020-113738GB-I00.



\bibliography{references}


\end{document}